\documentclass[twoside,12pt,reqno]{amsart}
\usepackage[foot]{amsaddr}
\usepackage[unicode=true,pdfusetitle, bookmarks=true,bookmarksnumbered=false,bookmarksopen=false, breaklinks=false,backref=false,colorlinks=false]{hyperref}
\hypersetup{colorlinks,linkcolor=myurlcolor,citecolor=myurlcolor,urlcolor=myurlcolor}
\usepackage{graphics,epstopdf,graphicx,amsthm,amsmath,amssymb,mathptmx,braket,colortbl,color,bm,framed,mathrsfs}
\usepackage[T1]{fontenc}
\usepackage[up]{subfigure}
\usepackage{tikz}
\usepackage{float}
\usepackage{amsfonts}
\usepackage{cite}
\usepackage{hyperref}
\hypersetup{
     colorlinks = true,
     linkcolor = black!60!blue,
     anchorcolor = brown,
     citecolor = purple,
     filecolor = brown,
     urlcolor = purple
}
\usepackage{orcidlink}

\usepackage[left=4cm, right=4cm, top=2cm]{geometry}

\usepackage[normalem]{ulem}

\usepackage{xcolor,cancel}

\def\mod{\mathrm{mod}\ }
\def\gap{\mathrm{gap}}

\def\bbC{\mathbb C}

\def\bbZ{\mathbb Z}

\def\modtwo{\mathrm{mod}_2}

\usepackage{circledsteps}

\usepackage{multirow,tabu}
\usepackage{array}
\usepackage{tabularx}
\newcolumntype{L}[1]{>{\raggedright\let\newline\\\arraybackslash\hspace{0pt}}m{#1}}
\newcolumntype{C}[1]{>{\centering\let\newline\\\arraybackslash\hspace{0pt}}m{#1}}
\newcolumntype{R}[1]{>{\raggedleft\let\newline\\\arraybackslash\hspace{0pt}}m{#1}}
\usepackage{tabu}

\usepackage{qcircuit}

\definecolor{myurlcolor}{rgb}{0,0,0.9}
\newcommand{\be}{\begin{equation}}
\newcommand{\ee}{\end{equation}}
\newcommand{\beq}{\begin{eqnarray}}
\newcommand{\eeq}{\end{eqnarray}}
\newcommand{\beqs}{\begin{eqnarray*}}
\newcommand{\eeqs}{\end{eqnarray*}}

\theoremstyle{plain}
\newtheorem{thm}{Theorem}
\newtheorem{lem}[thm]{Lemma}
\newtheorem{prop}[thm]{Proposition}
\newtheorem{cor}[thm]{Corollary}
\theoremstyle{definition}
\newtheorem{defn}[thm]{Definition}

\tikzstyle WL=[line width=10pt,opacity=1.0]
\newcommand*{\myproofname}{Proof}

\def\ot{\otimes}
\def\complex{\mathbb{C}}

\def \diag {\mathrm{diag}}

\def\sharpP{\# \mathsf P}
\def\poly{\mathrm{poly}}
\def\i{\mathrm{i}}

\makeatother

\begin{document}

\title{Classical simulation of quantum circuits by half Gauss sums}
\author{Kaifeng Bu\texorpdfstring{$^{* \dagger}$}{}}
\email{kfbu@fas.harvard.edu (K.Bu)}
\author{Dax Enshan Koh\orcidlink{0000-0002-8968-591X}\texorpdfstring{$^{\ddagger\mathparagraph\mathsection}$}{}}
\email{dax\_koh@ihpc.a-star.edu.sg (D.E.Koh)}

\address[$\dagger$]{School of Mathematical Sciences, Zhejiang University, Hangzhou, Zhejiang 310027, China}
\address[*]{Department of Physics, Harvard University, Cambridge, Massachusetts 02138, USA}

\address[$\ddagger$]{Department of Mathematics, Massachusetts Institute of Technology, Cambridge, Massachusetts 02139, USA}

\address[$\mathparagraph$]{
Zapata Computing, Inc., 100 Federal Street, 20th Floor, Boston, Massachusetts 02110, USA}

\address[$\mathsection$]{Institute of High Performance Computing,
Agency for Science, Technology and Research (A*STAR), 1
Fusionopolis Way, \#16-16 Connexis, Singapore 138632, Singapore}

\begin{abstract}

We give an efficient algorithm to evaluate a certain class of exponential sums, namely the periodic, quadratic, multivariate half Gauss sums.
We show that these exponential sums become $\sharpP$-hard to compute when we omit either the periodicity or quadraticity condition. 
We apply our results about these exponential sums to the classical simulation
of quantum circuits, and give an alternative proof of the Gottesman-Knill theorem.
We also explore a connection between these exponential sums and the Holant framework. In particular, we generalize the existing definition of affine signatures to arbitrary dimensions, and use our results about half Gauss sums to show that the Holant problem for the set of affine signatures is tractable.

\end{abstract}

\maketitle

\setcounter{tocdepth}{1}
\tableofcontents
\section{Introduction}

Exponential sums have been extensively studied  in number theory \cite{hua2012introduction} and have a rich history that dates back to the time of Gauss \cite{gauss1801disquisitiones}. They have found numerous applications in communication theory \cite{paterson1999application}, graph theory \cite{goldberg2010complexity}, coding theory \cite{shparlinski2002exponential,hurt1997exponential}, cryptography \cite{shparlinski2002exponential,shparlinski2002exponential_1}, algorithms \cite{shparlinski2002exponential} and many other areas of applied mathematics. 

More recently, they have also found useful applications in quantum computation. In 2005, Dawson et al.\ showed, using Feynman's sum-over-paths technique \cite{feynman2010quantum}, that the amplitudes of quantum circuits with Toffoli and Hadamard gates can be expressed in terms of exponential sums \cite{dawson2005quantum}. Such an approach has complexity-theoretic applications. For example, by noting that the exponential sum can be expressed as a $\mathsf{GapP}$-function, it can be used to show that the complexity class $\mathsf{BQP}$ is contained in $\mathsf{PP}$, a result first proved by \cite{adleman1997quantum} using different methods.

The idea of using exponential sums to express quantum amplitudes has been developed further in a number of subsequent works \cite{bacon2008analyzing,penney2017quantum, montanaro2017quantum,koh2017computing,amy2018controlled,amy2018towards, kocia2018stationary}. For example, in \cite{bacon2008analyzing}, Bacon, van Dam and Russell find an exponential-sum representation of the amplitudes of algebraic quantum circuits. They then exploit the theory of exponential sums to prove several properties of such circuits. For instance, they prove that in the limit of large qudit degree, the acceptance probabilities of such circuits converge to either zero or one.

The use of exponential sums to express quantum amplitudes elucidates a correspondence between quantum circuits and low-degree polynomials, called the \textit{circuit-polynomial correspondence} \cite{montanaro2017quantum}. This correspondence allows results about polynomials to be used to prove results about quantum circuits, and vice versa. For example, this correspondence was exploited in the forward direction by \cite{koh2017computing}, which provided an alternative proof of the Gottesman-Knill Theorem \cite{gottesman1997heisenberg} for 
quopit Clifford circuits, i.e.\ Clifford circuits in odd prime dimensions \cite{koh2017computing}, by showing that the amplitudes of such circuits can be expressed in terms of tractable exponential sums.

More generally, the circuit-polynomial correspondence also establishes a connection between exponential sums and the strong classical simulation of quantum circuits---deciding whether a class of quantum circuits is classically simulable, in many cases, can be reduced to the problem of deciding whether an exponential sum is tractable. This has important applications, for example, to the goal of quantum computational supremacy \cite{preskill2012quantum, harrow2017quantum,dalzell2018many}---the intractability of an exponential sum can be used to show that the class of circuits it corresponds to cannot be efficiently simulated.

In this paper, we consider a generalization of the exponential sums used in the above examples. In particular, we introduce the periodic, quadratic, multivariate half Gauss sum, and show that these incomplete Gauss sums can be computed efficiently using number-theoretic techniques. Moreover, we show that these exponential sums can be used to express the amplitudes of qudit Clifford circuits, thereby providing an alternative proof of the Gottesman-Knill theorem for qudit Clifford circuits. We also show that without the periodicity or quadraticity condition, these exponential  sums become intractable, under plausible complexity assumptions.

Our work improves on existing results in a number of ways.
First, while the results of \cite{montanaro2017quantum} and \cite{koh2017computing} are restricted to qubit and quopit systems, respectively, our results hold for all $d$-level systems. In doing so, we address a limitation of the approach used in \cite{koh2017computing}, where the proof of the Gottesman-Knill theorem works only for $d$-level systems, where $d$ is restricted to be an odd prime. Second, while previous works on tractable exponential sums are based on Gauss sums \cite{cai2010tractable, koh2017computing, lidl1997finite}, ours are based on half Gauss sums, which are a generalization of Gauss sums. Consequently, we find a larger class of tractable exponential sums compared to previous works. 
Third, we generalize the existing definition of affine signatures \cite{cai2010tractable} to arbitrary dimensions, and use our results about half Gauss sums to show that the Holant problem for the set of affine signatures is tractable. Fourth, we demonstrate the importance of a periodicity condition, which has not been previously explored, to the classical simulation of quantum circuits.

The rest of the paper is structured as follows. In Section \ref{sec:Results}, we summarize the main results of our work. In Section \ref{sec:FHGuas}, we define half Gauss sums and give an efficient classical algorithm to compute a subclass of these sums, namely the periodic, quadratic, multivariate half Gauss sums.
In Section \ref{sec:Cliffordcircuit}, we apply our results about half Gauss sums to Clifford circuits, and provide an alternative proof of the Gottesman-Knill Theorem.
In Section \ref{sec:hardnessresults}, we study the hardness of evaluating half Gauss sums that do not satisfy either the periodicity condition or the quadraticity condition. In Section \ref{sec:tractable}, we explore a connection between half Gauss sums and the Holant framework. We generalize the existing definition of affine signatures to arbitrary dimensions, and use our results about half Gauss sums to show that the Holant problem for the set of affine signatures is tractable.

\subsection{Our results}
\label{sec:Results}

The complexity of evaluating the exponential sum 
\begin{eqnarray}
Z(d, f)=\sum_{x_1,\ldots,x_n\in \bbZ_d}
\omega^{f(x_1,\ldots,x_n)}_d,
\end{eqnarray}
where $d, n \in \bbZ^+$ are positive integers, $\omega_d=\exp(2\pi \i/d)$ is a $d$th root of unit, 
and $f(x_1,\ldots,x_n)$ is a polynomial with integer coefficients, has been studied in previous works. In particular, it was proved that
$Z(d, f)$ can be evaluated in $\poly(n)$ time when $f$ is a  quadratic polynomial. This was first proved for the case when $d$ is a prime number \cite{lidl1997finite}, before being generalized to the case when $d$ is an arbitrary positive integer \cite{cai2010tractable}. On the other hand, when $f$ is a polynomial of degree $\geq 3$, the problem of evaluating such exponential sums was proved to be $\sharpP$-hard \cite{ehrenfeucht1990computational, cai2010tractable}. 

In this paper, we consider the following generalization of the above exponential sum:
\begin{eqnarray}
Z_{1/2}(d, f)=\sum_{x_1,\ldots,x_n\in\bbZ_d}
\xi^{f(x_1,\ldots,x_n)}_d.
\end{eqnarray}
Here, $\xi_d$ is a chosen square root of $\omega_d$ (i.e. $\xi_d^2 = \omega_d$)  satisfying $\xi^{d^2}_d=1$. 

Unlike $Z(d, f)$, the sum $Z_{1/2}(d, f) $ may not be evaluable in $\poly(n)$ time even when $f$ is a  
quadratic polynomial---the properties of the coefficients of the quadratic polynomial $f$ are crucial to determining the efficiency of evaluating $Z_{1/2}(d, f)$. Assuming plausible complexity assumptions, we prove that a necessary and sufficient condition to guarantee the efficiency of evaluating $Z_{1/2}(d, f)$ for quadratic polynomials $f$ is a periodicity condition, which states that 
\begin{eqnarray}
\label{eq:periodicityCondition}
\xi^{f(x_1,\ldots,x_n)}_d=\xi^{f(x_1 (\mod d),\ldots,x_n (\mod d))}_d,
\end{eqnarray}
for all variables $x_1,\ldots, x_n\in \bbZ$.
More precisely, we prove that for quadratic polynomials $f$ satisfying the periodicity condition, $Z_{1/2}(d, f)$ can be evaluated in $\poly(n)$ time, and that without the periodicity condition, there is no efficient algorithm to evaluate $Z_{1/2}$ unless the widely-believed assumption that $\mathsf{FP} \neq \sharpP$ is false. 
This is summarized by our main theorem: 

\begin{thm} \label{thm:mainTheorem} (Restatement of Theorems \ref{thm:mainThmOne}, \ref{thm:HardnessDegree3Periodic} and \ref{thm:HardnessDegree2Aperiodic})
Let $f\in \bbZ[x_1, \ldots, x_n]$ be a quadratic polynomial
over $n$ variables $x_1, \ldots, x_n$  satisfying the periodicity condition. Then
$Z_{1/2}(d, f)$ can be computed in polynomial time. If either the quadraticity or periodicity condition is omitted, then $Z_{1/2}(d, f)$ is $\sharpP$-hard to compute.
\end{thm}

\renewcommand{\arraystretch}{1.4}
\begin{table}
\begin{center}
\begin{tabu}{|C{2cm}|C{1.8cm}|[1.3pt] C{2.2cm}| C{2.2cm} | C{2.2cm}|}
\hline
\multicolumn{2}{|c|[1.3pt]}{$Z_{1/2^k}(2, f)$}  &  $\mathrm{deg}(f)=1$ & $\mathrm{deg}(f)=2$ & $\mathrm{deg}(f)\geq 3$  \\ \tabucline[1.3pt]{1-5}
{ periodic } &  $k\geq0$ & \cellcolor{blue!15} $\mathsf{FP}$ & \cellcolor{blue!15} $\mathsf{FP}$ & \cellcolor{red!25} $\#\mathsf{P}$-hard \\  \hline
aperiodic &  $k\geq 1$ & \cellcolor{blue!15} $\mathsf{FP}$ & \cellcolor{red!25} $\#\mathsf{P}$-hard & \cellcolor{red!25} $\#\mathsf{P}$-hard \\  \hline
\end{tabu}
\end{center}
\caption{\small Hardness of computing $Z_{1/2^k}(2,f)$, where $k\geq 0$ or $k\geq 1$, and $f$ is a polynomial function with coefficients in $\bbZ$ and domain $\bbZ_2^n$. Here, `periodic' means that $f$ satisfies the periodicity condition \eqref{eq:periodicityCondition}, and `aperiodic' means that $f$ does not necessarily satisfy it. The label $\mathsf{FP}$ means that $Z_{1/2^k}(d,f)$ can be computed in classical polynomial time, and $\#\mathsf{P}$-hard means that there is no efficient classical algorithm to compute $Z_{1/2^k}(d,f)$, unless the widely-believed conjecture $\mathsf{FP}\neq\#\mathsf{P}$ is false.
}
\label{tab:classification}
\end{table}

We consider
the case $d=2$, and study the complexities of evaluating more general exponential sums, namely those of the form:
\begin{equation}
    Z_{1/2^k}(2, f)=\sum_{x_1,...,x_n\in \bbZ_2}\omega^{f(x_1,\ldots,x_n)}_{2^{k+1}},
\end{equation} where $k\geq 0$ is an integer and $f$ is a polynomial with $n$ variables.
Our classification results are summarized in Table \ref{tab:classification}.

Next, we apply Theorem \ref{thm:mainTheorem} to the classical simulation of Clifford circuits. In particular, we show that 
the output probabilities of Clifford circuits can be expressed in terms of half Gauss sums:
\begin{thm}
\label{thm:GKtheoremResults}
(Simplified version of Theorem \ref{thm:GKtheorem})
Let $C$ be an $m$-qudit Clifford circuit. Let $a \in \bbZ_d^m$ and $b \in \bbZ_d^k$. Then the probability of obtaining the outcome $b$ when the first $k$ qudits of $C\ket a$ are measured is given by
\begin{equation}
P(b|a) := || \bra b_{1..k} C \ket a_{a..m} ||^2 = \frac{1}{d^{l}} Z_{1/2}(d, \phi),
\end{equation}
where $l\in \bbZ$ and
$\phi$ is a quadratic polynomial that satisfies the periodicity condition \eqref{eq:periodicityCondition}. Moreover, $l$ and $\phi$ can be computed efficiently.
\end{thm}

Since half Gauss sums can be computed efficiently, 
Theorem \ref{thm:GKtheoremResults} implies that there is an efficient strong simulation of Clifford circuits. This gives an alternative proof (which does not make use of stabilizer techniques) of the Gottesman-Knill Theorem \cite{gottesman1997heisenberg}.

\section{Half Gauss sums}\label{sec:FHGuas}

\subsection{Univariate case}
Given two nonzero integers $a, d$ with $d>0$ and $\gcd(a, d)=1$, the Gauss sum\footnote{also referred to as the ``univariate quadratic homogeneous Gauss sum". See Appendix \ref{sec:terminology}.}
\cite{Lang1970Gsum} is defined as:
\begin{eqnarray}
\label{eq:univariateGauss}
G(a, d)=\sum_{x\in \bbZ_d}\omega_d^{ax^2},
\end{eqnarray}
where $\omega_d=\exp(2\pi \i/d)$ is a root of unity. It has been proved that the Gauss sum $G(a, d)$ can be computed in polynomial time in $\log a$ and $\log d$ \cite{Lang1970Gsum}. Several useful properties of the Gauss sum $G(a, d)$ have been provided in 
Appendix \ref{append:GS}.

In this section, we define a generalization of the Gauss sum, called the half Gauss sum\footnote{also referred to as the ``univariate quadratic homogeneous half Gauss sum''. See Appendix \ref{sec:terminology}. Also, note that our definition of ``half Gauss sum'' differs from that used in \cite{berndt1980half}.}: given two nonzero integers $a, d$ with $d>0$ and $\gcd(a, d)=1$, let
\begin{eqnarray}
\label{eq:univariateHalfGauss}
G_{1/2}(a, d)=\sum_{x\in \bbZ_d}\xi^{ax^2}_d.
\end{eqnarray}
Here, $\xi_d$ is a chosen square root of $\omega_d$ such that 
$\xi^{d^2}_d=1$. This condition is chosen so that the summation over the ring $\bbZ_d$ is well-defined, i.e. if $x \equiv y \  (\mod d)$, then $\xi_d^{ax^2} = \xi_d^{ay^2}$.  Note that such a condition on $\xi_d$ has also been used in the investigation of reflection positivity in parafermion algebra to ensure that the twisted product is well-defined \cite{jaffe2017reflection,jaffe2017planar}.

For $d=1$, we get $G_{1/2}(a, 1)=1$, which is trivial; hence, we will subsequently restrict our attention to the nontrivial case of
$d\geq 2$. Note that we have two choices for $\xi_d$ when $d$ is even, namely (i) $\xi_d=\omega_{2d}$ for all even $d$, and (ii) 
$\xi_d= -\omega_{2d}$ for all even $d$. Since the analyses in both cases are similar, we will present only the first case in this section, and refer the reader to Appendix \ref{appen:xi_ev} for the second case. In other words, $\xi_d$ may be expressed as follows:
\begin{align}
    \xi_d= \begin{cases}
    -\omega_{2d}=\omega^{(d+1)/2}_d , & \mbox{$d$ odd} \\
    \omega_{2d} , & \mbox{$d$ even}.
    \end{cases}
\end{align}

We will now present properties of the half Gauss sum, its relationship with the
Gauss sum, and the computational complexity of evaluating the half Gauss sum.

\begin{prop}\label{thm:HG1}
The half Gauss sum satisfies the following properties:

\begin{enumerate}
    \item If $d$ is odd, then
\begin{eqnarray}
G_{1/2}(a, d)=G(a(d+1)/2, d).
\end{eqnarray}
\item If d is even, then 
\begin{eqnarray}
G_{1/2}(a, d)=G_{1/2}(a(N_1+bN_2),b)G_{1/2}(aN_2, c),
\end{eqnarray}
where $d=bc$, $\gcd(b,c)=1$,  $2|b$, and $N_1$ and $N_2$ are integers satisfying $N_1c+N_2b=1$.
\end{enumerate}
\end{prop}
  
\begin{proof}
\hfill
\begin{enumerate} 
    \item  If $d$ is odd,  
$\gcd((d+1)/2, d)=1$ and 
$\gcd(a,d)=1$. Thus, we have $\gcd(a(d+1)/2, d)=1$. Therefore, 
we have 
\begin{eqnarray*}
G_{1/2}(a, d)=\sum_{x\in \bbZ_d}\xi^{ax^2}_d
=\sum_{x\in \bbZ_d}\omega^{a\frac{d+1}{2} x^2}_d
=G(a(d+1)/2, d).
\end{eqnarray*}

\item If $d$ is even, then $a$ must be odd since $\gcd(a, d)=1$. Hence,
\begin{eqnarray*}
G_{1/2}(a, d)=\sum_{x\in \bbZ_d}\xi^{ax^2}_d
=\sum_{x\in \bbZ_d} \omega^{ax^2}_{2d}.
\end{eqnarray*}  

Moreover, $d$ can be decomposed as 
$d= bc$ with $\gcd(b, c)=1$. 
Since $d$ is even, it follows that one of $b$ and $c$ is divisible by 2. Without loss of generality, we assume that
 $2|b$, which implies that $c\equiv 1 \ (\mod 2)$.
Since $\gcd(b, c)=1$, 
 there exist two integers
$N_1$ and $N_2$ such that $N_1c+N_2b=1$. 
By the Chinese remainder theorem, there exists an isomorphism
$\bbZ_d\to \bbZ_b\times \bbZ_c$ : $x\mapsto (y,z)$ with 
$x\equiv y \ (\mod b)$ and $x\equiv z \ (\mod c)$. In fact, we can choose the map 
$x=N_2bz+N_1cy$, which can also be written as 
\begin{eqnarray*}
x=y+N_2b(z-y)
=z+N_1c(y-z).
\end{eqnarray*}
Thus,
\begin{eqnarray*}
\omega^{ax^2}_{2d}
=\omega^{aN_1x^2}_{2b}
\omega^{aN_2x^2}_{2c}.
\end{eqnarray*}
Moreover, 
\begin{eqnarray*}
\omega^{aN_1x^2}_{2b}
=\omega^{aN_1[y^2+2bN_2(z-y)+N^2_2b^2(y-z)^2]}_{2b}
=\omega^{aN_1y^2}_{2b},
\end{eqnarray*}
where the last equality comes from the fact that $2|b$,
and
\begin{eqnarray*}
\omega^{aN_2x^2}_{2c}
&=&\omega^{aN_2[z^2+2N_1c(y-z)+N^2_1c^2(y-z)^2]}_{2c}\\
&=&\omega^{aN_2z^2}_{2c}
\omega^{aN_2N^2_1c^2(y-z)^2}_{2c}\\
&=&\omega^{aN_2z^2}_{2c}
\omega^{aN_2N^2_1c^2(y^2+z^2)}_{2c}.
\end{eqnarray*}
Since $\omega^{c^2}_{2c}=(-1)^c=-1$ and 
$N_1$ is odd as  $N_2b+N_1c=1$, we have 
\begin{eqnarray*}
\omega^{aN_2x^2}_{2c}
=\omega^{aN_2z^2}_{2c}
(-1)^{aN_2(y^2+z^2)}
&=&(-\omega_{2c})^{aN_2z^2}
(-1)^{aN_2y^2}\\
&=&\xi_{c}^{aN_2z^2}
(-1)^{aN_2y^2}.
\end{eqnarray*}
Thus,
\begin{eqnarray*}
\omega^{ax^2}_{2d}
=\omega^{aN_1y^2}_{2b}
\xi_{c}^{aN_2z^2}
(-1)^{aN_2y^2}
&=&\omega^{a(N_1+bN_2)y^2}_{2b}
\xi_{c}^{aN_2z^2}\\
&=&\xi^{a(N_1+bN_2)y^2}_{b}
\xi_{c}^{aN_2z^2}.
\end{eqnarray*}
Since
$c(N_1+bN_2)+b(1-c)N_2=1$, it follows that 
$\gcd(N_1+bN_2, b)=1$. Thus
$\gcd(a(N_1+bN_2), b)=1$. But
 $\gcd(aN_2, c)=1$. Therefore, we have 
 \begin{eqnarray*}
 G_{1/2}(a, d)
 &=&\sum_{y\in \bbZ_b, 
 z\in \bbZ_c}
 \xi^{a(N_1+bN_2)y^2}_{b}
\xi_{c}^{aN_2z^2}\\
&=&G_{1/2}(a(N_1+bN_2), b)
G_{1/2}(aN_2, c).
 \end{eqnarray*}
 
\end{enumerate}
 \end{proof}

Now, any even number $d$ can always be decomposed into 
$d=2^mc$ with $m\geq 1$  and $c$ being odd. It is straightforward to see that 
  \begin{eqnarray*}
 G_{1/2}(a, d)
=G_{1/2}(a(N_1+2^mN_2), 2^m)
G_{1/2}(aN_2, c),
\end{eqnarray*}
 where $N_22^m+N_1c=1$.
 As $c$ is odd,  it can be rewritten as 
 a Gauss sum by Proposition \ref{thm:HG1}. And so it remains for us to evaluate the half Gauss sum for
 $d=2^m$, i.e., $G_{1/2}(a, 2^m)$. 
 
 \begin{prop}
 If $m\geq 3$, then 
 \begin{eqnarray}
 G_{1/2}(a, 2^m)
 = 2G_{1/2}(a, 2^{m-2}).
 \end{eqnarray}
 Moreover, 
 \begin{eqnarray}
 G_{1/2}(a, 2)
 =1+\i^{a},\\
 G_{1/2}(a, 2^2)
 =2\omega^a_8.
 \end{eqnarray}

  \end{prop}
 \begin{proof}
 First,  $G_{1/2}(a, 2)$ and $ G_{1/2}(a, 2^2)$ can be obtained by direct calculation. 
 
 Second, for $m\geq 3$,
 \begin{eqnarray*}
 G_{1/2}(a, 2^{m})&=&
 \sum_{x\in [2^m]}
 \omega^{ax^2}_{2^{m+1}}\\
& =&\sum_{x\in [2^{m-1}]}
 \left[\omega^{ax^2}_{2^{m+1}} + \omega^{a(x+2^{m-1})^2}_{2^{m+1}} \right]\\
 &=&\sum_{x\in [2^{m-1}]}\omega^{ax^2}_{2^{m+1}} 
 \left[1+ \omega^{a2^mx+a2^{2m-2}}_{2^{m+1}}\right]\\
 &=&\sum_{x\in [2^{m-1}]}\omega^{ax^2}_{2^{m+1}} 
 \left[1+(-1)^{x}\right]\\
 &=& \sum_{y\in [2^{m-2}]}
 \omega^{a(2y)^2}_{2^{m+1}} [1+(-1)^{2y}]\\
&=&2 \sum_{y\in [2^{m-2}]}
 \omega^{4ay^2}_{2^{m+1}} 
 =2 \sum_{y\in [2^{m-2}]}
 \omega^{ay^2}_{2^{m-1}} \\
 &=&2G_{1/2}(a, 2^{m-2}).
 \end{eqnarray*}

  \end{proof}

Based on the above properties of the half Gauss sum $G_{1/2}(\cdot, \cdot)$ and the fact that the Gauss sum
$G(\cdot, \cdot)$ can be calculated in $\poly(\log a, \log d)$-time, we obtain the following corollary:

\begin{cor}
Given two nonzero integers $a, d$ with $d>0$ and $\gcd(a,d)=1$, the half Gauss sum can be calculated in 
$\poly(\log a, \log d )$ time.

\end{cor}

\subsection{Multivariate case}

\label{sec:multivariate}

In this section, we consider a generalization of the Gauss sum \eqref{eq:univariateGauss} to the multivariate case:
 \begin{eqnarray}
 \label{eq:Zdf}
 Z(d, f)
 =\sum_{x_1,\ldots,x_n\in \bbZ_d}
 \omega^{f(x_1,\ldots, x_n)}_d,
 \end{eqnarray}
where each $x_i$ is summed over a finite 
ring $\bbZ_d$, and $f(x_1,\ldots, x_n)$ is a quadratic polynomial with integer coefficients. The  \textit{multivariate quadratic Gauss sum} \eqref{eq:Zdf} has been proved to be evaluable in polynomial time \cite{cai2010tractable}.

We also consider an analogous multivariate generalization of the half Gauss sum:
\begin{eqnarray}
\label{eq:Z12df}
Z_{1/2}(d, f)
=\sum_{x_1,...,x_n\in \bbZ_d }
\xi^{f(x_1,...,x_n)}_d,
\end{eqnarray}
where 
$f(x_1,..., x_n)=\sum_{i\leq j \in[n]}\alpha_{ij}x_ix_j+\sum_{i\in[n]}\beta_ix_i+\gamma_0$ is a quadratic polynomial 
with integer coefficients. However, $Z_{1/2}(d, f)$ may not be efficiently evaluable even
for quadratic polynomials. It turns out that the existence of an efficient algorithm depends on 
some periodicity condition.

We say that a polynomial $f$ satisfies the \textit{periodicity condition}\footnote{More generally, we say that a function $g: \bbZ^n\to \complex$ is \textit{periodic} with period $d$ if
\begin{eqnarray}
g(x_1,\ldots,x_n)=g(x_1 (\mod d),\ldots,x_n (\mod d))
\end{eqnarray}
for all variables $x_1,\ldots,x_n\in \bbZ$. 
} if 
\begin{eqnarray}
\xi^{f(x_1,...,x_n)}_d=\xi^{f(x_1 (\mod d),...,x_n (\mod d))}_d,
\end{eqnarray}
for all  variables $x_1,..., x_n\in \bbZ$. This periodicity condition can also be regarded as the well-definedness condition of $Z_{1/2}$
on $\bbZ_d$. If $d$ is an odd number, then 
$\xi_d=-\omega_{2d}$, i.e,  $\xi^d_d=1$, which implies that the periodicity condition can always be satisfied for odd $d$. 
However, the periodicity condition may not be satisfied in the case of even $d$. 

\begin{prop}
Let $d$ be even, and let  $f(x_1,\ldots, x_n)=\sum_{i\leq j \in[n]}\alpha_{ij}x_ix_j+\sum_{i\in[n]}\beta_ix_i+\gamma_0$, 
be a quadratic polynomial. Then, $f$ satisfies the periodicity condition if and only if
 the cross terms $\alpha_{ij}$ ($i<j$) and
linear terms  $\beta_i$ are all even.
\end{prop}
\begin{proof}
It is easy to verify that the quadratic polynomial $f$ satisfies the periodicity condition if all
the cross terms $\alpha_{ij}$ ($i<j$) and 
linear terms  $\beta_i$ are even.

In the other direction, if $f$ satisfies the periodicity condition, 
then $\xi^{f(x_1,...,x_n)}_d=\xi^{f(x_1 (\mod d),...,x_n (\mod d))}_d$ for any 
$x_1,...,x_n\in \bbZ$. Thus,  for any $i$, 
\begin{eqnarray*}
\xi^{\alpha_{ii}x^2_i+\beta_i x_i}_d
=\xi^{\alpha_{ii}(x_i+d)^2+\beta_i (x_i+d)}_d ,
\end{eqnarray*}
for any $x_i\in \bbZ$
by choosing $x_j=0$ for any $j\neq i$.
Besides, $\xi_d $ satisfies the conditions $\xi^{2d}_d=1$ and $\xi^{d^2}_d=1$.
Thus, $\xi^{\beta_i d}_d=(-1)^{\beta_i}=1$, which implies that $\beta_i$ is an even number. 
Since $i$ was chosen arbitrarily, all linear terms $\beta_i$ are even.
Besides, for any fixed $i$ and  $j$ with $i<j$, we can choose $x_k=0$ for any $k\neq i, j$:
\begin{eqnarray*}
\xi^{\alpha_{ii}x^2_i+\alpha_{jj}x^2_j+\alpha_{ij}x_ix_j+\beta_i x_i+\beta_jx_j}_d
=\xi^{\alpha_{ii}(x_i+d)^2+\alpha_{jj}x^2_j+\alpha_{ij}(x_i+d)x_j+\beta_i (x_i+d)+\beta_jx_j}_d ,
\end{eqnarray*}
for any $x_i, x_j\in \bbZ$. This implies that $\alpha_{ij}$ is even. 
Since $i$ and $j$ were arbitrarily chosen, all the cross terms $\alpha_{ij}$ are even.

\end{proof}

The periodicity condition of  the polynomial $f$ plays an important in the efficient evaluation of the 
exponential sum $Z_{1/2}$. We denote the set of quadratic polynomials satisfying the periodicity condition by $\mathcal F_2^{\mathrm{p.c.}}$.
For any quadratic polynomial $f$ satisfying this periodicity condition, the exponential sum $Z_{1/2}(d, f)$ can be evaluated in polynomial 
time given the description of $f$.

\begin{thm}
\label{thm:mainThmOne}

If $f \in \mathcal F_2^{\mathrm{p.c.}}$ is a quadratic polynomial satisfying the periodicity condition, then 
$Z_{1/2}(d, f)$ can be evaluated in polynomial time.
\end{thm}
\begin{proof}
Consider the expression 
\begin{eqnarray*}
f(x_1,..., x_n)=\sum_{i\leq j \in[n]}\alpha_{ij}x_ix_j+\sum_{i\in[n]}\beta_ix_i+\gamma_0,
\end{eqnarray*}
with the cross term $\alpha_{ij}$ ($i<j$) and linear term $\beta_i$ being even.
We may assume that $\gamma_0=0$, as it only contributes an additive constant term to $Z_{1/2}(d,f)$.

\noindent\underline{Case (i)}: All diagonal terms $\alpha_{ii}$ are even. In this case,
$Z_{1/2}(d, f)=Z(d, f/2)$, which can be evaluated in polynomial time \cite{cai2010tractable}.

\noindent\underline{Case (ii)}: There exists at least one diagonal term $\alpha_{ii}$ that is odd. 

\noindent\underline{Case (iia)}:
 $d$ is odd. Then, $\xi_d=\omega^{(d+1)/2}_d$. Thus, 
$Z_{1/2}(d, f)=Z(d, \frac{d+1}{2}f)$, which can be evaluated in polynomial time \cite{cai2010tractable}.

\noindent\underline{Case (iib)}:  $d=2^m$. Then, $\xi_d=\omega_{2d}$. 
Since there exists at least one diagonal term $\alpha_{ii}$ that is odd, we assume 
that $\alpha_{11}$ is odd without loss of generality.
 Since $\alpha_{11}$ is odd, it is invertible in 
$\bbZ_{2d}$ with $2d=2^{m+1}$. We can rewrite the quadratic polynomial $f$ to separate the term involving $x_{1}$: 
\begin{eqnarray*}
f(x_1,\ldots, x_n)
=\alpha_{11}[x^2_1
+x_1f_1(\hat{x}_1, x_2, \ldots, x_n)]
+f_2(\hat{x}_1, x_2,...,x_n),
\end{eqnarray*}
where $f_1$ is a linear function over $n-1$  variables
$\set{x_2,\ldots,x_n}$ with 
\begin{eqnarray*}
f_1(\hat{x}_1, x_2,\ldots,x_n)
=\sum_{j\geq 2}\alpha^{-1}_{11}\alpha_{1j} x_j+\alpha^{-1}_{11}\beta_1,
\end{eqnarray*}
and  $f_2$ is a quadratic polynomial with even cross terms and linear terms over $n-1$  variables $\set{x_2,...,x_n}$. Here, the notation $\hat{x}_1$ means that the variable $x_1$ is absent from the polynomial.

Since the cross terms and linear terms are even,
\begin{eqnarray*}
f_1=2f'_1=2\left(\sum_{j\geq 2}\frac{\alpha^{-1}_{11}\alpha_{1j}}{2} x_j+\frac{\alpha^{-1}_{11}\beta_1}{2}\right).
\end{eqnarray*}
Thus,
\begin{eqnarray*}
f=\alpha_{11}(x_1+f'_1)^2+f',
\end{eqnarray*}
where $f'$ is a quadratic polynomial with even cross terms and linear terms over $n-1$  variables $\set{x_2,...,x_n}$.
Therefore,
\begin{eqnarray*}
Z_{1/2}(d, f)
=\sum_{x_1,...,x_n\in \bbZ_d}
\xi^{\alpha_{11}(x_1+f'_1)^2+f'}_d
&=&\sum_{x_2,..,x_n\in \bbZ_d}\xi^{f'}_d \sum_{x_1\in \bbZ_d}\xi^{\alpha_{11}(x_1+f'_1)^2}_d\\
&=&Z_{1/2}(d, f')G_{1/2}(\alpha_{11}, d),
\end{eqnarray*}
where the last equality comes from the fact that the summation over 
$x_1\in \bbZ_d$ is independent of the value of $f'_1$. 
This reduces the
evaluation  of $Z_{1/2}(d, f)$ to 
$Z_{1/2}(d, f')$
where $f'$ is a quadratic polynomial over $n -1$ variables with 
even cross terms and linear terms. We can repeat this step 
until all the diagonal terms are even, 
  which then reduces to Case (i).

\noindent\underline{Case (iic)}: $d=2^mc$, with $c $ being odd and $c\geq3$. Then, $\xi_d=\omega_{2d}$. 
Since there exists at least one diagonal term $\alpha_{ii}$ that is odd, we shall take,
without loss of generality, the first $t$ diagonal terms
$\alpha_{ii}$ ($1\leq i \leq t$) to be odd and the other diagonal terms
$\alpha_{ii}$ ($i\geq t+1$) to be even.

Now, we can rewrite $f$ as follows
\begin{eqnarray*}
f(x_1,...,x_n)
=\sum^{t}_{i=1}x^2_{i}
+f_1(x_1,..,x_n),
\end{eqnarray*}
where the coefficients of the quadratic form $f_1$ are all even. 
Hence, $f=\sum^{t}_{i=1}x^2_{i}+2f'_1$, with $f'_1=f_1/2$.

Since $\gcd(2^m, c)=1$, 
 there exist two integers
$N_1$ and $N_2$ such that $N_22^m+N_1c=1$. 
Adopting a process similar to that used in the proof of Proposition \ref{thm:HG1}, we find, using the Chinese remainder theorem, that there exists an isomorphism
$\bbZ_d\to \bbZ_{2^m}\times \bbZ_c$ :: $x_i\mapsto (y_i, z_i)$ with 
$x_i\equiv y_i \ (\mod 2^m)$ and $x_i\equiv z_i \ (\mod c)$. Thus, we have
\begin{eqnarray*}
&&Z_{1/2}(d, f)\\
&=&\sum_{x_1,...,x_n\in \bbZ_d}
\xi^{\sum^t_{i=1}x^2_i}_d\omega^{f'_1(x_1,..,x_n)}_{d}\\
&=&\sum_{y_1,...,y_n\in \bbZ_{2^m}}
\sum_{z_1,...,z_n\in \bbZ_c}
\xi^{\sum^t_{i=1}(N_1+2^mN_2)y^2_i}_{2^m}\xi^{\sum^t_{i=1}N_2z^2_i}_c
\omega^{N_1f'_1(y_1,...,y_n)}_{2^m}
\omega^{N_2f'_1(z_1,...,z_n)}_{c}\\
&=&\sum_{y_1,...,y_n\in \bbZ_{2^m}}\xi^{\sum^t_{i=1}(N_1+2^mN_2)y^2_i}_{2^m}\omega^{N_1f'_1(y_1,...,y_n)}_{2^m}
\sum_{z_1,...,z_n\in \bbZ_c}\xi^{\sum^t_{i=1}N_2z^2_i}_c\omega^{N_2f'_1(z_1,...,z_n)}_{c}\\
&=&\sum_{y_1,...,y_n\in \bbZ_{2^m}}\xi^{\sum^t_{i=1}(N_1+2^mN_2)y^2_i}_{2^m}\omega^{(N_1+2^mN_2)f'_1(y_1,...,y_n)}_{2^m} \\
&& \times
\sum_{z_1,...,z_n\in \bbZ_c}\xi^{\sum^t_{i=1}N_2z^2_i}_c\omega^{N_2f'_1(z_1,\ldots,z_n)}_{c}\\
&=&Z_{1/2}(2^m, (N_1+2^mN_2)f)
Z_{1/2}(c, N_2f),
\end{eqnarray*}
where the second-to-last equality comes from the fact that $\omega^{2^{m}}_{2^{m}}=1$.
This reduces the computation of $Z_{1/2}(d, f)$ to Case (iia) and Case (iib).

\end{proof}

Here, we have shown the existence of efficient algorithms to evaluate half Gauss sums with quadratic polynomials that satisfy the periodicity condition. We note, however, that if we omit either the periodicity or quadraticity condition, these sums become hard to compute (under a plausible complexity-theoretic conjecture). We will return to a discussion of this in Section \ref{sec:hardnessresults}.

Finally, we note here that there is a  nice relationship between half Gauss sums $Z_{1/2}(d,f)$ and the number of zeros of functions of the form $f(x)-k \ (\mod d)$ or $(\mod 2d)$. We explore this further in Appendix \ref{app:zeros}.

\section{\texorpdfstring{$m$}{m}-qudit Clifford circuits}
\label{sec:Cliffordcircuit}

In this section, we apply our results on the half Gauss sum to Clifford circuits. Let $d\geq 2$ and $m\geq 1$ be integers. The $m$-qudit Clifford group is the set of operations (called \textit{Clifford operations}) on $m$ qudits that are generated by the following gates: $X, Y, Z, F, G, CZ$ \cite{farinholt2014ideal,jaffe2017planar,jaffe2017constructive,jaffe2018holographic}. 

Here, $X, Y$ and $Z$ are the $d$-level Pauli matrices defined by
\begin{eqnarray} \label{eq:PauliMatrices}
X\ket{k}=\ket{k+1},\ Y\ket{k}=\xi^{1-2k}_d\ket{k-1},\ Z\ket{k}=\omega^k_d\ket{k},
\end{eqnarray}
$F$ is the Fourier gate defined by
\begin{equation}
\label{eq:Fouriermatrix}
    F\ket{k}=\frac{1}{\sqrt{d}}\sum^{d-1}_{l=0}\omega^{kl}_d\ket{l},
\end{equation}
$G$ is the Gaussian gate defined by
\begin{equation}
\label{eq:Gaussianmatrix}
G\ket{k}=\xi^{k^2}_d\ket{k},
\end{equation}
and $CZ$ is the controlled-$Z$ gate defined by
\begin{equation}
\label{eq:CZmatrix}
CZ\ket{k_1,k_2}=\omega^{k_1k_2}_d\ket{k_1,k_2}.
\end{equation}
Note that the gates $X, Y, Z$ are the qudit generalizations of the qubit Pauli gates \begin{equation}
    \sigma_x = \begin{pmatrix} 0 & 1 \\ 1 & 0
    \end{pmatrix},\ \sigma_y = \begin{pmatrix} 0 & -\i \\ \i & 0
    \end{pmatrix},\ \sigma_z = \begin{pmatrix} 1 & 0 \\ 0 & -1
    \end{pmatrix},
\end{equation}
and the $F$, $G$ and $CZ$ gates are the qudit generalizations of
the Hadamard gate $\tfrac 1{\sqrt 2}(X+Z)$, the phase gate $\diag(1,\i)$, and the controlled-$Z$ gate $\diag(1,1,1,-1)$, respectively, on qubits. 

It is straightforward to check that the gates \eqref{eq:PauliMatrices}--\eqref{eq:CZmatrix} satisfy the following algebraic relations \cite{jaffe2017constructive,jaffe2017planar}:
\begin{eqnarray*}
X^d=Y^d=Z^d=F^4=G^{2d}=(FG)^3q_d^{-1}=I,\\
XYX^{-1}Y^{-1}=YZY^{-1}Z^{-1}=ZXZ^{-1}X^{-1}=\omega_d,\\
XYZ=\xi_d,\ FXF^{-1}=Z,\ GXG^{-1}=Y^{-1},
\end{eqnarray*}
where 
$$q_d=\frac{1}{\sqrt{d}}\sum^{d-1}_{j=0}\xi^{j^2}_d.$$
From the above identities, it is easy to see that the $X$ and $Y$ gates can be expressed in terms of the other gates, and so the following gate set suffices to generate the Clifford group:
$\mathcal C = \{Z, G, F, CZ\}$. An $m$-\textit{qudit Clifford circuit} is a circuit with $m$ registers and whose gates are all Clifford operations. We shall assume that the Clifford circuit is unitary, i.e. there are no intermediate measurements in the circuit\footnote{Note that the results in this section do not hold if the Clifford circuit contains intermediate measurements whose outcomes affect which gates or measurements are performed next. These circuits are called adaptive Clifford circuits, and their amplitudes are $\sharpP$-hard to compute in general \cite{jozsa2014classical, koh2015further}.}.

Without loss of generality, we will assume that (i) each register of the Clifford circuit $C$ begins with an $F$ gate and ends with an $F^\dag$ gate, and that (ii) the internal circuit (i.e. the full circuit minus the first and last layers) consists of only gates in $\mathcal C$.
In other words, $C$ is of the form
\begin{equation}
    C = (F^\dag)^{\otimes m} C' F^{\otimes m},
\end{equation}
where the internal circuit $C'$ comprises only gates in $\mathcal C$.
This loses no generality because any Clifford circuit can be transformed into a circuit of the above form, first, by inserting 4 $F$ gates at the start of each register and the pair $F^\dag F$ at the end of each register, and second, by compiling the internal circuit using only gates in $\mathcal C$.

For each $m$-qudit Clifford circuit, we adopt the following labeling scheme: divide each horizontal wire of the internal part of $C$ into segments, with each segment corresponding to a portion of the wire which is either between 2 $F$ gates, or between an $F$ gate and an $F^\dag$ gate. It is easy to verify that the total number of segments is given by $n = h-m$, where $h$ is the total number of $F$ or $F^\dag$ gates (including those in the first and last layers) in $C$. Label the segments $x_1, \ldots, x_n$.

We will also use the following terminology. The leftmost labels on each register are called \textit{inceptive indices}. The rightmost labels on each register are called \textit{terminal indices}. All other indices are called \textit{internal indices}. For a set of indices $I = \{i_1,\ldots, i_s\}$, we use $x_I$ to denote the tuple $(x_{i_1},\ldots,x_{i_s})$.

\begin{defn}
\label{def:phase_polynomial}
Let $C$ be a Clifford circuit with labels $\{x_1,\ldots, x_n\}$. The \textit{phase polynomial}\footnote{This definition is chosen specifically so that both Proposition \ref{prop:Cliffordperiodic} and  Eq.~\eqref{eq:amplitudesAsHalfGaussSum}, which will be stated later, hold. Note that there are examples of Clifford circuits $C$ and phase polynomials $S_C$ for which Eq.~\eqref{eq:amplitudesAsHalfGaussSum} holds but  Proposition \ref{prop:Cliffordperiodic} does not. For example, consider the single-qubit Clifford circuit $HSH$, where $S = \sum_{x \in \set{0,1}} i^x\ket{x}\!\bra{x}$ is the phase gate. Since $x=x^2$ for all $x\in \mathbb Z_2$, the all-zero amplitude of $C$ can be written as a half-Gauss sum in two different ways:
\begin{align}
\bra 0 C\ket 0 = Z_{1/2}(2, S ) =  Z_{1/2}(2,S')
\end{align}
where $S(x)=x$ and $S'(x) = x^2$. While $S'$ satisfies the periodicity condition (for $d=2$), $S$ does not. } of $C$ is the polynomial 
\begin{equation}
\label{eq:phasePolynomial}
S_C(x_1,\ldots,x_n) = 2 \sum_{\gamma \in \Gamma} \prod_{i \in I_\gamma} x_i + \sum_{g \in \mathcal G} \prod_{j \in I_g}  x_j^2 ,
\end{equation}
where $\Gamma$ is the set of internal $F, Z, CZ$ gates, and $\mathcal G$ is the set of $G$ gates in $C$.
\end{defn}

We now show that if $C$ is a Clifford circuit, then its phase polynomial $S_C$ is a quadratic polynomial that satisfies the periodicity condition.
\begin{prop}
\label{prop:Cliffordperiodic}
If $C$ is a Clifford circuit, then $S_C \in \mathcal F_2^{\mathrm{p.c.}}$.
\end{prop}
\begin{proof}
Since each gate in $C$ is incident to at most 2 segments, the degree of the polynomial is at most 2. The only terms which can have odd coefficients are terms of the form $x_i^2$. The remaining terms, which are all either linear and cross terms, have even coefficients, which implies that $S_C \in \mathcal F_2^{\mathrm{p.c.}}$.
\end{proof}

The reverse direction is also true: for every polynomial $S \in \mathcal F_2^{\mathrm{p.c.}}$, there exists a Clifford circuit $C$ such that $S = S_C$, as the following proposition shows:

\begin{prop}
Let $\mathcal A$ be the class of Clifford circuits. The function 
\begin{eqnarray}
\Theta : \mathcal A & \rightarrow & \mathcal F_2^{\mathrm{p.c.}} \\
C &\mapsto & S_C 
\end{eqnarray}
is surjective.
\end{prop}
\begin{proof}
Let 
\begin{equation*}
S = \sum_{i\leq j \in [n]} \alpha_{ij} x_i x_j + \sum_{i \in [n]} \beta_i x_i \in \mathcal F_2^{\mathrm{p.c.}} ,
\end{equation*}
i.e. $\alpha_{ij}$ is even for $i<j$ and $\beta_i$ is even for all $i$.
Construct the circuit $C = (F^\dag)^{\ot n} C' F^{\ot n}$, where $C'$ is defined as follows:
\begin{enumerate}
\item for each $i \in [n]$, apply the gate $G$ $\alpha_{ii}$ times.
\item for each $i < j \in [n]$, apply the gate $CZ$ $\alpha_{ij}/2$ times.
\item for each $i \in [n]$, apply the gate $Z$ $\beta_i/2$ times.
\end{enumerate}
Then,
\begin{equation*}
S_C = \sum_{i \in [n]} \alpha_{ii} x_i^2 + 2\left(\sum_{i<j \in [n]} \tfrac{\alpha_{ij}}{2} x_{ij} + \sum_{i\in [n]} \tfrac{\beta_i}{2} x_i \right) = S,
\end{equation*}
which implies that $\Theta$ is surjective.
\end{proof}

We now show that the amplitudes of Clifford circuits can be expressed in terms of half Gauss sums.

\begin{thm}
\label{thm:zeroAmplitudes}
Let $C = (F^\dag)^{\ot m} C' F^{\ot m}$ be an $m$-qubit Clifford circuit with $h$ $F$ or $F^\dag$ gates and $n=h-m$ labels $x_1,\ldots,x_n$. Then,
\begin{eqnarray}
\label{eq:amplitudesAsHalfGaussSum}
\bra{0}^{\ot m }C \ket{0}^{\ot m} 
=\frac{1}{\sqrt{d^h}}
\sum_{x_1,\ldots , x_n \in \bbZ_d}
\xi^{S_C(x_1,\ldots ,x_n)}_d
=\frac{1}{\sqrt{d^h}} Z_{1/2}(d, S_C).
\end{eqnarray}
\end{thm}
\begin{proof}
Apply the sum-over-paths technique \cite{dawson2005quantum,koh2017computing} to the Clifford circuit $C$.
\end{proof}

Theorem \ref{thm:zeroAmplitudes} can be easily generalized to also allow us to compute amplitudes of Clifford circuits with arbitrary computational-basis states as inputs or outputs:
\begin{prop}
Let $C = (F^\dag)^{\ot m} C' F^{\ot m}$ be an $m$-qudit Clifford circuit with $h$ $F$ or $F^\dag$ gates and $n=h-m$ labels $x_1,\ldots,x_n$. Let $a,b\in \bbZ_d^m$.
Then,
\begin{eqnarray}
\bra b C \ket a 
=\frac{1}{\sqrt{d^h}} Z_{1/2}(d, S_C+ 2 a\cdot x_I + 2 b\cdot x_F),
\end{eqnarray}
where $I$ and $J$ are the inceptive and terminal indices (written in order) of $C$ respectively.

\end{prop}
\begin{proof}
We start by writing
\begin{align*}
\bra b (F^\dag)^{\ot m} C' F^{\ot m} \ket a &= \bra {0^m} (X^\dag)^b (F^\dag)^{\ot m} C' F^{\ot m} X^a \ket {0^m} \nonumber\\
&= \bra {0^m} (F^\dag)^{\ot m} (Z^\dag)^b C' Z^a F^{\ot m}  \ket {0^m} .
\end{align*}

Note that $C^* = (F^\dag)^{\ot m} (Z^\dag)^b C' Z^a F^{\ot m}$ is itself a Clifford circuit, and we could apply Theorem \ref{thm:zeroAmplitudes} to it:
\begin{eqnarray*}
\bra b C \ket a 
=\frac{1}{\sqrt{d^h}} Z_{1/2}(d, S_{C^*}),
\end{eqnarray*}
where
\begin{eqnarray*}
S_{C^*}(x_1,\ldots ,x_n) = S_c(x_1,\ldots, x_n) +  2 a\cdot x_I + 2 b\cdot x_F .
\end{eqnarray*}
\end{proof}

A corollary of the above result is that we can express the probabilities of outcomes of qudit Clifford circuits in terms of half Gauss sums even when only a subset of registers is measured. This was previously shown to hold for quopit Clifford circuits \cite{Priv2017Penney}, i.e., qudit Clifford circuits, where $d$ is an odd prime.

\begin{thm}
\label{thm:GKtheorem}
Let $C = (F^\dag)^{\ot m} C' F^{\ot m}$ be an $m$-qudit Clifford circuit with $h$ $F$ or $F^\dag$ gates and $n=h-m$ labels $x_1,\ldots,x_n$. Assume that $C'$ contains at least one $F$ gate on each register. Let $I$ be the inceptive indices, $J$ be the internal indices, $F$ be the first $k$ terminal indices, and $E$ be the last $m-k$ terminal indices. Let $a \in \bbZ_d^m$ and $b \in \bbZ_d^k$. Then the probability 
\begin{equation}
P(b|a) = || \bra b_{1..k} C \ket a_{a..m} ||^2
\end{equation}
of obtaining the outcome $b$ when the first $k$ qudits of $C\ket a$ are measured is given by
\begin{equation}
P(b|a) = \frac{1}{d^{n+k}} Z_{1/2}(d, \phi),
\end{equation}
where
\begin{eqnarray}
\phi(x_I, y_I, x_F, y_F, x_J, y_J, w_E) &=& S_c(x_I,x_J,x_F,w_E) - S_c(y_I,y_J,y_F,w_E)  \nonumber\\ &&\quad + 2a \cdot (x_I - y_I) + 2 b \cdot (x_F-y_F).
\end{eqnarray}
\end{thm}

\begin{proof}
\begin{eqnarray}
P(b|a) &=& || \bra b_{1..k} U \ket a_{a..m} ||^2 \nonumber\\
&=& \sum_{\beta\in \bbZ_d^{m-k}} \left|\bra{b\beta} C \ket a\right|^2 \nonumber\\
&=& \sum_{\beta\in \bbZ_d^{m-k}} \left| \frac 1{\sqrt h} Z_{1/2} (d,S_C + 2a\cdot x_I + 2(b,\beta)\cdot(x_F,x_E)
\right|^2 \nonumber\\
&=& \frac 1{d^h} \sum_{x,y\in \bbZ_d^n} \xi_d^{S_C(x)-S_C(y)+2a\cdot(x_I-y_I)+2b\cdot(x_F-y_F)} \sum_{\beta \in \bbZ_d^{m-k}} \omega_d^{\beta\cdot(x_E-y_E)} \nonumber\\
&=& \frac 1{d^{h-m+k}} \sum_{x_I,y_I\in \bbZ_d^n}\sum_{x_F,y_F\in \bbZ_d^k}\sum_{x_J,y_J\in \bbZ_d^{n-2m}} \sum_{w_E \in \bbZ_d^{m-k}} 
\xi_d^{\phi(x_I, y_I, x_F, y_F, x_J, y_J, w_E)} \nonumber\\
&=&
\frac{1}{d^{n+k}} Z_{1/2}(d, \phi).
\end{eqnarray}
where in the fifth line, we used the property that
\begin{equation}
    \sum_{\beta \in \bbZ_d^{m-k}} \omega_d^{\beta\cdot(x_E-y_E)} = d^{m-k} \delta_{x_E,y_E}.
\end{equation}
\end{proof}

Since half Gauss sums can be computed efficiently, the above proof gives an alternative proof of the Gottesman-Knill Theorem \cite{gottesman1997heisenberg} for all qudit Clifford circuits:
\begin{cor}
\label{cor:GKstrong}
{\normalfont (Gottesman-Knill Theorem---strong version)}
Qudit Clifford circuits acting on computational basis input states can be efficiently simulated (in the strong sense \cite{nest2010classical}) by a classical computer.
\end{cor}

Since strong simulation implies weak simulation \cite{terhal2004adptive}, Corollary \ref{cor:GKstrong} implies that there is an efficient classical algorithm that samples from the output distributions of qudit Clifford circuits. Note that such an efficient classical simulation algorithm exists even in the case  when there is a logarithmic number of $T$ gates \cite{bravyi2016improved}.

\section{Hardness results and complexity dichotomy theorems}
\label{sec:hardnessresults}

In this section, we show that extending the class of periodic quadratic half Gauss sums in various ways leads to intractable exponential sums. See Table \ref{tab:classification} for a summary of our results.

\subsection{Degree-3 polynomials with periodicity condition}
\label{sec:deg3}

We shall show, under plausible complexity assumptions, that if we omit the quadraticity condition (while possibly keeping the periodicity condition) from Theorem \ref{thm:mainTheorem}, then there is no efficient algorithm that can compute the exponential sum $Z_{1/2}(d,f)$ on all inputs $(d,f)$. More formally, consider the following problem.

\begin{tabular}{ccp{10.3cm}}
($\mathcal A$) 
&
\texttt{Input}: 
&
$f$, where $f: \mathbb Z^n \rightarrow \mathbb Z$ is a polynomial function of degree $\leq 3$ that satisfies the periodicity condition
\\
&
\texttt{Output}:
&
$Z_{1/2}(2,f) = \sum_{x \in \mathbb Z_2^n} \i^{f(x)}$.
\end{tabular}
\\ \\
Our goal is to show that $(\mathcal A)$ is $\#\mathsf P$-hard to compute. To this end, we consider the following problem.

\begin{tabular}{ccp{10.3cm}}
($\mathcal B$) 
&
\texttt{Input}: 
&
$g$, where $g: \mathbb Z^n \rightarrow \mathbb Z$ is a polynomial of degree $\leq 3$
\\
&
\texttt{Output}:
&
$\gap(g) = \sum_{x \in \mathbb Z_2^n} (-1)^{g(x)}$.
\end{tabular}
\\ \\
It is well-known that ($\mathcal B$) is a $\#\mathsf P$-hard problem (see Theorem 1 of \cite{ehrenfeucht1990computational}). Hence, to show that ($\mathcal A$) is also $\#\mathsf P$-hard, it suffices to show that there is an efficient reduction from ($\mathcal B$) to ($\mathcal A$). Indeed, such a reduction is provided by the following chain of equalities:
\begin{align}
    \gap(g) = \sum_{x \in \mathbb Z_2^n} (-1)^{g(x)} = \sum_{x \in \mathbb Z_2^n} \i^{2 g(x)} = Z_{1/2}(2,2g).
\end{align}
Since $2g$ satisfies the periodicity condition\footnote{This can be verified directly by using the definition of periodicity. Alternatively, this also follows immediately from Theorem \ref{thm:degree3characterization} in Appendix \ref{sec:periodicPoly3d2}, where we fully characterize the set of periodic polynomials with degree $\leq 3$ when $d=2$.} for $d=2$, it follows that $\gap(g)$ can be efficiently computed given an efficient algorithm for $\mathcal A$.

Combining these results with Theorem \ref{thm:mainThmOne} gives the following theorem.
\begin{thm}
\label{thm:HardnessDegree3Periodic}
The following computational problem is $\#\mathsf P$-hard:

\begin{tabular}{ccp{10.3cm}}
($\mathcal C$) 
&
$\mathtt{Input}$: 
&
$(d,f)$, where $d \in \mathbb Z_{\geq 2}$ and $f: \mathbb Z^n \rightarrow \mathbb Z$ is a degree-3 polynomial function that satisfies the periodicity condition
\\
&
$\mathtt{Output}$:
&
$Z_{1/2}(d,f) = \sum_{x \in \mathbb Z_2^n} \i^{f(x)}$.
\end{tabular}

\end{thm}

\subsection{Degree-2 polynomials without periodicity condition}
\label{subsubsec:without}

We shall show, under plausible complexity assumptions, that if we omit the periodicity condition (while keeping the quadraticity condition) from Theorem \ref{thm:mainTheorem}, then there is no efficient algorithm that can compute the exponential sum $Z_{1/2}(d,f)$ on all inputs $(d,f)$.

To see this, we first consider the following problem:

\begin{tabular}{ccp{10.3cm}}
($\mathcal D$) 
&
\texttt{Input}: 
&
$f$, where $f: \mathbb Z^n \rightarrow \mathbb Z$ is a polynomial function of degree $\leq 2$
\\
&
\texttt{Output}:
&
$Z_{1/2}(2,f) = \sum_{x \in \mathbb Z_2^n} \i^{f(x)}$.
\end{tabular}
\\ \\
Note that the inputs of $(\mathcal D)$ are allowed to be any arbitrary polynomial of degree $\leq 2$, including those that do not satisfy the periodicity condition. We will now show that $(\mathcal B)$ reduces to $(\mathcal D)$.

\begin{thm}
There exists a polynomial-time reduction from $(\mathcal B)$ to $(\mathcal D)$.
\end{thm}
\begin{proof}
Assume that there exists an oracle $O_{\mathcal D}$ for the problem $(\mathcal D)$. We will use it to construct a polynomial-time algorithm $T_{\mathcal B}$ for $(\mathcal B)$ as follows. Let $g$ denote the input to the algorithm $T_{\mathcal B}$, i.e. $g:\mathbb Z^n \rightarrow \mathbb Z$ is a polynomial of degree $\leq 3$. For $1\leq i < j < k \leq n$, let 
$a_{ijk}, a_{ij}, a_{i}, a \in \mathbb Z_2$ be the coefficients of the polynomial $g \ (\mod 2)$, viz.
\begin{align}
    g(x_1,\ldots,x_n) &= \sum_{1\leq i_1<i_2<i_3\leq n} 
    a_{i_1,i_2,i_3} x_{i_1} x_{i_2} x_{i_3}
    +
    \sum_{1\leq i_1<i_2\leq n} 
    a_{i_1,i_2} x_{i_1} x_{i_2}
    \nonumber\\
    &\quad +
    \sum_{i=1}^n a_i x_i 
    + a
    \ (\mod 2).
\end{align}

Note that the ability to represent $g \ (\mod 2)$ as a multilinear polynomial arises from the identity $x^2 = x$ for $ x \in \mathbb Z_2$. The motivation for expressing $g$ in the above form comes from the fact that the desired output $\gap(g)$ of $T_{\mathcal B}$ depends on only values $g(x) (\mod 2)$.

Next, we exploit the circuit-polynomial correspondence \cite{montanaro2017quantum} to construct an IQP circuit $C$ over the gate set $\set{Z,CZ,CCZ}$ whose circuit amplitudes can be expressed in terms of the gap of $g$. Let $C = H^{\otimes n} C' H^{\otimes n}$ be an IQP circuit whose internal circuit $C'$ is constructed as follows:
\begin{enumerate}
    \item[(i)] Place a $Z$ gate on the $i$th wire if $a_i=1$.
    \item[(ii)] Place a $CZ$ gate between the $i$th and $j$th wires if $a_{ij}=1$.
    \item[(iii)] Place a $CCZ$ gate between the $i$th, $j$th and $k$th wires if $a_{ijk}=1$. 
\end{enumerate}

Then, the amplitude of measuring the all-zero string when the circuit $C$ is applied to the all-zero state is given by
\begin{align}
\label{eq:amplitudeGap1}
    \bra 0 C \ket 0 = \frac{1}{2^n}\gap(g-a) = \frac{1}{2^n}\sum_{x\in \mathbb Z_2^n} (-1)^{g(x)-a} = \frac{1}{2^n} (-1)^a \gap(g).
\end{align}

Now, construct the circuit $C_{\mathcal G}$ that performs the same unitary operation as $C$, but which consists of only gates in $\mathcal G$, where $\mathcal G$ is the strictly universal\footnote{Note that $Z$ is not needed for universality, since $\set{H,CS}$ is already universal (see \cite{kitaev1997quantum} or Theorem 1 of \cite{aharonov2003simple}).} gate set $\mathcal G = \set{H,Z,CS}$, where $CS = \diag(1,1,1,\i)$ is the controlled-phase gate satisfying
$CS\ket{x_i,x_j} = \i^{x_i x_j} \ket{x_i,x_j}$. To achieve this, we
replace all the $CZ$ and $CCZ$ gates in $C$ by circuit gadgets comprising only $H$ and $CS$ gates. This may be achieved by making use of the following circuit identity (which follows from Lemma 6.1 of \cite{barenco1995elementary}):
\begin{align}
\Qcircuit @C=1em @R=0.3em @!R {
 & \ctrl 1 & \qw  &&&& & \qw & \ctrl 1 & \qw & \ctrl 1 & \ctrl 2 & \qw \\
 & \ctrl 1 & \qw &&=&& & \ctrl 1 & \targ & \ctrl 1 & \targ & \qw & \qw\\
 & \control \qw & \qw &&&& &
 \gate S & \qw & \gate{S^\dag} & \qw & \gate S & \qw
}
\end{align}
as well as the following identities:
\begin{align}
    CZ &= (CS)^2, \\
    C(S^\dag) &= (CS)^3,\\
    CX_{12} &= H_2 CZ_{12} H_2,
\end{align}
which allow the gates $CCZ$ and $CZ$ to be expressed completely in terms of $H$ and $CS$. Note that by construction, each register in $C_{\mathcal G}$ begins and ends with a $H$ gate, i.e. $C_{\mathcal G} = H^{\otimes n} C_{\mathcal G}' H^{\otimes n}$ for some circuit $C_{\mathcal G}'$ over the gate set $\mathcal G$.

Next, mirroring the labeling scheme for Clifford circuits described in Section \ref{sec:Cliffordcircuit}, we construct the phase polynomial $f$ corresponding to $C_{\mathcal G}$ as follows:
\begin{enumerate}
    \item[] Divide each wire of the internal circuit $C_{\mathcal G}'$ into segments, with each segment corresponding to a portion of the wire between two $H$ gates. Label the segments $x_1,\ldots, x_N$, where the total number of segments is $N := h-n$, where $h$ is the total number of $H$ gates in $C_{\mathcal G}$.
\end{enumerate}

Define the phase polynomial of $C_{\mathcal G}$ to be \begin{align}
    f(x_1,\ldots,x_N) = 2 \sum_{\gamma \in \Gamma} \prod_{I\in I_\gamma} x_i + \sum_{g\in \mathcal G} \prod_{i \in I_g} x_i,
\end{align}
where $\Gamma$ is the set of internal $H$ gates and $\mathcal G$ is the set of $CS$ gates.

Then, the following all-zero amplitude of $C_{\mathcal G}$ may be written as
\begin{align}
\label{eq:amplitudeGap2}
    \bra 0 C_{\mathcal G} \ket 0 = \frac{1}{\sqrt h} Z_{1/2}(2,f) = \frac{1}{\sqrt h} \sum_{x\in \mathbb Z_2^N} \i^{f(x)}.
\end{align}

Since $\bra 0 C_{\mathcal G} \ket 0 = \bra 0 C \ket 0$, it follows from Eqs.~\eqref{eq:amplitudeGap1} and \eqref{eq:amplitudeGap2} that 
\begin{align}
\label{eq:FormulaForGap}
    \gap(g) = \frac{2^n}{\sqrt h} (-1)^a Z_{1/2}(2,f).
\end{align}

Next, feed $f$ into the oracle $O_{\mathcal D}$ to get $Z_{1/2}(2,f)$. Finally, use Eq.~\eqref{eq:FormulaForGap} to calculate and output $\gap(g)$.

Since each step of the above reduction $T_{\mathcal B}$ takes polynomial time, the entire reduction runs in polynomial time.

\end{proof}

Since $(\mathcal B)$ is $\#\mathsf P$-hard, it follows from the above reduction that $(\mathcal D)$ is also $\#\mathsf P$-hard. Combining this results with Theorem \ref{thm:mainThmOne} gives the following theorem.
\begin{thm}
\label{thm:HardnessDegree2Aperiodic}
The following computational problem is $\#\mathsf P$-hard:

\begin{tabular}{ccp{10.3cm}}
($\mathcal C$) 
&
$\mathtt{Input}$: 
&
$(d,f)$, where $d \in \mathbb Z_{\geq 2}$ and $f: \mathbb Z^n \rightarrow \mathbb Z$ is an aperiodic degree-2 polynomial function
\\
&
$\mathtt{Output}$:
&
$Z_{1/2}(d,f) = \sum_{x \in \mathbb Z_2^n} \i^{f(x)}$.
\end{tabular}

\end{thm}

\subsection{Other incomplete Gauss sums: }

In this section, we restrict our attention to $d=2$, and consider incomplete Gauss sums of the form: 
\begin{equation}
    Z_{1/2^{k}}(2, f)=\sum_{x_1,...,x_n\in \bbZ_2}\omega^{f(x_1,...,x_n)}_{2^{k+1}}
\end{equation}
with $k\geq 2$.
For $k=2$,
the exponential sum $$Z_{1/4}=\sum_{x_1,...,x_n\in \bbZ_2}\omega^{f(x_1,...,x_n)}_{8} ,$$ 
with no requirement on the periodicity of the 
polynomial $f$,
corresponds to the gate set $\{H,T, CZ\}$, which is universal, and it can be shown that computing such sums is $\#\mathsf P$-hard.
However, for quadratic polynomial $f$ satisfying the periodicity condition, we can reduce the evaluation of $Z_{1/4}(2, f)$ to the evaluation of $Z_{1/2}(2, f')$, for some quadratic polynomial 
$f'$ satisfying the periodicity condition, which in turn can be evaluated in $\poly(n)$ time. 
More generally, for any $k\geq 2$, if $f$ is a quadratic polynomial satisfying the
periodicity condition,
the incomplete Gauss sum $Z_{1/2^k}(2, f)$ can be reduced to $Z_{1/2}(2, f')$.

\begin{lem}
Let $d=2$, and let $f=\sum_{i\leq j}\alpha_{ij} x_ix_j+\sum_i\beta_i x_i$ be a quadratic polynomial. Then $f$ satisfies the periodicity condition 
 \begin{eqnarray}
\omega^{f(x_1,...,x_n)}_{2^{k+1}} 
=\omega^{ f((x_1 \mod 2),...,(x_n \mod 2))}_{2^{k+1}},
\end{eqnarray}
if and only if  $2^{k-1}|\alpha_{ii}$, $2^{k}|\alpha_{ij}$ (i< j) and 
$2^k|\beta_i$. Thus, 
$Z_{1/2}^{k+1}(2, f)=Z_{1/2}(2, f/2^{k-1})$, where $f/2^{k-1} $ satisfies the periodicity condition
for $\omega_2=\sqrt{-1}$.

\end{lem}
\begin{proof}
It is easy to verify that the quadratic polynomial $f$ satisfies the periodicity condition if 
 $2|\alpha_{ii}$, $4|\alpha_{ij}$ ($i< j$) and 
$4|\beta_i$. 

For any $i$, 
\begin{eqnarray*}
\omega^{\alpha_{ii}x^2_i+\beta_i x_i}_{2^{k+1}}
=\omega^{\alpha_{ii}(x_i+2)^2+\beta_i (x_i+2)}_{2^{k+1}}
\end{eqnarray*}
for any $x_i\in \bbZ$, which implies that $2^{k-1}|\alpha_{ii}$ and $2^k|\beta_i$.

Moreover, for any fixed $i$ and  $j$ with $i<j$, we can choose $x_k=0$ for any $k\neq i, j$ to get
\begin{eqnarray*}
\omega^{\alpha_{ii}x^2_i+\alpha_{ii}x^2_j+\alpha_{ij}x_ix_j+\beta_i x_i+\beta_jx_j}_{2^{k+1}}
=\omega^{\alpha_{ii}(x_i+2)^2+\alpha_{ii}x^2_j+\alpha_{ij}(x_i+2)x_j+\beta_i (x_i+2)+\beta_jx_j}_{2^{k+1}}
\end{eqnarray*}
for any $x_i, x_j\in \bbZ$. This implies that $2^k|\alpha_{ij}$. 
Since $i,j$ were arbitrarily chosen, it follows that all cross terms $\alpha_{ij}$ satisfy $2^k|\alpha_{ij}$.

\end{proof}

\subsection{Complexity dichotomy theorems}

In 1979, Valiant introduced the complexity class $\sharpP$ to characterize the computational complexity of
counting problems \cite{Valiant1979}, and ever since then, this has been a subject of much research.

Among the many important results arising from this research are the complexity dichotomy theorems, which have attracted considerable attention \cite{Creignou1996,Dyer2007,Bulatov2004,Goldberg2010,Bulatov2008,Dyer2009, cai2014complexity}.
These theorems state, roughly, that for certain classes of counting problems, each problem in the class is either efficiently computable or
$\sharpP$-hard. (See \cite{Cai2017book} for an overview.)

These dichotomy theorems have applications to the study of exponential sums. An example of such a theorem was provided by \cite{cai2010tractable}, which proved that computing Gauss sums $Z(d,f)$ can be performed efficiently when $\deg(f)\leq 2$ and is $\sharpP$-hard when $\deg(f) \geq 3$. Note that the polynomials considered by \cite{cai2010tractable} all satisfy the periodicity condition. Hence, if we combine these $\sharpP$-hardness results with Theorem \ref{thm:mainThmOne}, we arrive at a new dichotomy theorem: if $\deg(f)\leq 2$, then the exponential sum  $Z_{1/2}(d, f)$ is computable in polynomial time. Otherwise, if 
$\deg(f)\geq 3$, then computing $Z_{1/2}(d, f)$ is $\# \mathsf{P}$-hard.

Furthermore, for the class of aperiodic exponential sums, our results imply another new complexity dichotomy theorem: if $\deg(f)\leq 1$, then the exponential sum  $Z_{1/2}(d,f)$ is computable in polynomial time, otherwise if 
$\deg(f)\geq 2$, then computing $Z_{1/2}(d,f)$ is $\# \mathsf{P}$-hard. For a summary of these results, see Table \ref{tab:classification}.

\section{Tractable signature in Holant problem}
\label{sec:tractable}

In this section, we will apply our results about half Gauss sums to an important framework called the Holant framework, which we will now describe. Let $\mathcal{F}$ be a set of functions, where each element $f\in \mathcal{F}: \bbZ^{ n}_d \to \complex$. 
A signature grid 
$\Omega=(G, \mathcal{F})$ is a tuple, 
where $G=(V, E)$ is a hypergraph
and each $v\in \mathcal{F}$ is assigned a function 
$f_v\in F$ with arity equal to the 
number of hyperedges incident to it. 
A $\bbZ_d$ assignment $\sigma$
for every $e\in E$ gives an evaluation
$\prod_vf_v(\sigma|_{E(v)})$, where $E(v)$ 
denotes the edges incident to $v$.
Given an input instance $\Omega$, we are interested in computing
\begin{eqnarray}
\mathrm{Holant}_{\Omega}=\sum_{\sigma:E\to \bbZ_d}
\prod_v f_v(\sigma|_{E(v)}).
\end{eqnarray}

Affine signatures over $\mathbb Z_2$ and $\mathbb Z_3$ were defined in \cite{cai2014complexity, Williams2015}. In this section, we give a definition of affine signtures over $\mathbb Z_d$, for $d\geq 2$.

\begin{enumerate}
\item {\bf Affine signature over $\bbZ_d$}: Let $f$ be a signature of arity $n$ with inputs $x_1,..., x_n$ over the domain $\bbZ_d$, 
then $f$ is affine if it has the following form 
\begin{eqnarray}
\lambda\chi_{A\vec{x}=0}\xi^{g(x_1,..., x_n)}_d
\end{eqnarray}
where $\lambda\in \complex$, $\xi_d$ is a chosen square root of $\omega_d=\exp(2\pi i/d)$  such that 
$\xi^{d^2}_d=1$, $A$ is a matrix over $\bbZ_d$, $\chi$ is a $0$--$1$ indicator function such that 
$\chi_{A\vec{x}=0}=1$ if and only if $A\vec{x}=0$, and $g(x_1,,,.x_n)\in \bbZ[x_1,...,x_n]$ is a quadratic polynomial with 
even cross and linear terms. Let $\mathcal{A}$ to be the set of all affine signatures. It is straightforward to check that $\mathcal{A}$ is closed under multiplication.
\item {\bf Degenerate function on $n$ variables} Let 
\begin{eqnarray}
\mathcal{D}
=\set{\otimes_i[f_i(0), f_i(1),...,f_i(d-1)]| f_i(j)\in \complex}
\end{eqnarray}
 be the set of functions that can be expressed as the tensor product of 
unary function. 
\item {\bf The set $\mathcal{P}$}: Let $\mathcal{P}$ be the set of functions that can be written as the 
composition of unary functions and the binary equality relation $=_2$, where $=_2\!\!(i,j)$ is equal to $1$ if $i=j$ and $0$ otherwise. 
\end{enumerate}

\begin{thm}
\label{thm:holant}
Given a class of functions $\mathcal{F}$, if $\mathcal{F}\subseteq \mathcal{A}$ or $\mathcal{F}\subseteq \mathcal{P}$, then
$\mathrm{Holant}(\mathcal{F}$) is computable in polynomial time. 
\end{thm}
\begin{proof}
(1) If $\mathcal{F}\subseteq \mathcal{P}$, then following 
\cite{cai2014complexity}, we can group the variables into connected components if these variables are connected by the binary equality relation $=_2$. In any connected component, let us start with a variable that takes a value in $\bbZ_d$, 
and follow any edges labeled by the binary equality relation. There is at most one
extension of this assignment, i.e., each variable in this connected component must take the same value as the value that was taken at the beginning.
Then we can easily compute the value of the Holant by simply multiplying all the values. There are at most $d$ values, as we have $d$ choices at the starting edge.

(2) If $\mathcal{F}\subseteq \mathcal{A}$, then the method in \cite{cai2014complexity} may not work, as Gaussian elimination may not be applicable for general $\bbZ_d$. 
To get around this, we consider the inner product representation of the 
Holant problem $\mathrm{Holant}(\mathcal{F})$, which can be 
written as
\begin{eqnarray}
\mathrm{Holant}(\mathcal{F})
=(\otimes_e\bra{\mathrm{GHZ}_e})(\otimes_v\ket{f_v}),
\end{eqnarray}
where 
$\ket{\mathrm{GHZ}_e}$ denotes the GHZ state on $(\complex_d)^{\otimes |e|}$, where $|e|$ denotes the number of vertices incident to 
the edge $e$. For example, if $|e|\ =\{1,2,3\}$, then 
$\ket{\mathrm{GHZ}_e}$ is $\ket{+}=\sum^{d-1}_{i=0}\ket{i}$, 
$\ket{\mathrm{Bell}}=\sum^{d-1}_{i=0}\ket{ii}$ and 
$\ket{\mathrm{GHZ}}=\sum^{d-1}_{i=0}\ket{iii}$, respectively.

Since $f_v\in \mathcal{A}$,
\begin{eqnarray}
\ket{f_v}=\sum_{x_1,...,x_k\in \bbZ_d}\chi_{A_v\vec{x}=0}
\xi^{g_v(x_1,\ldots,x_k)}_d\ket{x_1,\ldots,x_k} ,
\end{eqnarray}
where $g_v$ is a quadratic polynomial with even cross and linear terms, and $k$ denotes the arity of $f_v$. If we omit the term $\chi_{A_v\vec{x}=0}$ in the above expression, then the remaining expression represents 
a stabilizer state, which we denote as $\ket{\mathrm{STAB}}_v$.
Now consider  $\sum^{k}_{i=1}A_{1,i}x_i+A_{1,k+1}=0\ (\mod d)$ that is given by the first line of $A\vec{x}=0$. We can add an ancilla qudit 
with $\bra{0}\prod_{j} (CX)^{A_{1j}} X^{A_{1, k+1}}\ket{0}$ with 
control qudit being $j=1,\ldots, k$. 
Then, $\ket{f_v}$ can be written as
\begin{eqnarray}
\ket{f_v}
=\bra{0}^{\otimes m_v} \prod_{i, j} (CX)^{A_{ij}} X^{A_{i, k+1}}\ket{\mathrm{STAB}}_v\ket{0}^{\otimes m_v},
\end{eqnarray}
where $m_v$ is the number of rows in $A_v$. 
Therefore,
\begin{eqnarray*}
\mathrm{Holant}(\mathcal{F})=(\otimes_e\bra{\mathrm{GHZ}_e})
(\otimes_v \bra{0}^{\otimes m_v})
(\otimes_v\prod_{i, j} (CX)^{A_{ij}} X^{A_{i, k+1}}\ket{\mathrm{STAB}}_v\ket{0}^{\otimes m_v}),
\end{eqnarray*}
which is just a product of two stabilizer states. It can be computed 
in polynomial time by the Gottesman-Knill theorem \cite{gottesman1997heisenberg}.

\end{proof}

While Theorem \ref{thm:holant} addresses the question about which functions lead to tractable Holant problems, we leave open the question about which functions lead to intractable Holant problems: for which classes of functions $\mathcal{F}$ does it hold that (i) $\mathcal F$ is neither in $\mathcal{P} $ nor $\mathcal{A}$ and (ii) $\mathrm{Holant}(\mathcal{F})$ is $\#\mathsf P$-hard?

\section{Concluding remarks}
In this paper, we found a larger (compared to previous results) class of quadratic exponential sums whose evaluation we proved to be tractable. In particular,
we studied the periodic, quadratic, multivariate half
Gauss sums, and gave an efficient algorithm to evaluate these incomplete Gauss sums.
We showed that without either the periodicity or quadraticity condition, these exponential 
sums become intractable under plausible complexity assumptions.
These results
demonstrate the importance of a periodicity condition, which
has not been explored in previous works.
Moreover, we show that these tractable
exponential sums can be used to express the amplitudes of qudit Clifford circuits,
thereby providing an alternative proof of the Gottesman-Knill theorem for qudit Clifford
circuits.
Last but not least, we provided a tractable affine signature in arbitrary dimensions in the Holant framework.

\section*{Acknowledgments}
We thank Arthur Jaffe for useful discussions. K.B. thanks the Templeton Religion Trust for partial support of this research under grant TRT0159, the ARO Grant W911NF-19-1-0302 and the ARO MURI Grant W911NF-20-1-0082, and also thanks Zhejiang University for the support of an Academic Award for
Outstanding Doctoral Candidates.
D.E.K.~was funded by EPiQC, an NSF Expedition
in Computing, under Grant CCF-1729369.

\begin{appendix}

\section{Exponential sum terminology}
\label{sec:terminology}

In this appendix, we summarize some of the terminology used in the main text. An \textit{exponential sum} is a sum of the form
\begin{equation}
    \sum_{x\in A} e^{f(x)},
\end{equation}
where $A \subseteq V$ is a finite set, $V$ is an arbitrary set, and $f:V\rightarrow \bbC$ is a complex-valued function.

The exponential sums used in this paper are all \textit{incomplete Gauss sums}\footnote{Here, we generalized the definition of ``incomplete Gauss sums'' used in \cite{lehmer1976incomplete,evans2003incomplete} to the multivariate case.}, which are sums of the form
\begin{equation}
Z_I(d,b,f) = \sum_{x_1,\ldots,x_n\in \bbZ_d}
\omega^{f(x_1,\ldots,x_n)}_b
\end{equation}
where $d,n,b \in \bbZ^+$ satisfy $d \leq b$ and $f$ is a polynomial with integer coefficients.

Two special cases of incomplete Gauss sums are the \textit{Gauss sum}, defined as
\begin{equation}
Z(d, f)= Z_I(d,d,f) = \sum_{x_1,\ldots,x_n\in \bbZ_d}
\omega^{f(x_1,\ldots,x_n)}_d.
\end{equation}
and the \textit{half Gauss sum}, defined as
\begin{equation}
Z_{1/2}(d, f) = \sum_{x_1,\ldots,x_n\in \bbZ_d}
\xi^{f(x_1,\ldots,x_n)}_d.
\end{equation}

With this terminology, note that Gauss sums are a special case of half Gauss sums, which are in turn a special case of incomplete Gauss sums.

When $f$ is quadratic, $Z(d,f)$ and $Z_{1/2}(d,f)$ reduce to the (multivariate) quadratic Gauss sum \eqref{eq:Zdf} and (multivariate) quadratic half Gauss sum \eqref{eq:Z12df} respectively. When $n=1$ and $f$ is a homogeneous quadratic polynomial (i.e. $f(x) = ax^2$), the sums $Z(d,f)$ and $Z_{1/2}(d,f)$ reduce to the univariate quadratic homogeneous Gauss sum \eqref{eq:univariateGauss} (which is usually just referred to as a Gauss sum \cite{Lang1970Gsum}) and univariate quadratic homogeneous half Gauss sum \eqref{eq:univariateHalfGauss} respectively. Note that univariate quadratic Gauss sums are also called Weil sums \cite{lidl1997finite}.

\section{Properties of  Gauss sum}\label{append:GS}
In this section, we give some basic facts about the Gauss sum $G(\cdot, \cdot)$\cite{Lang1970Gsum}.
Given two nonzero integers $a, d$ with $d>0$ and $\gcd(a, d)=1$,
\begin{eqnarray*}
G(a, d)
=\sum_{x\in \mathbb{Z}_d}
\omega^{ax^2}_d.
\end{eqnarray*}
The Gauss sum satisfies the following properties:

(1) If $d$ is odd, then 
\begin{eqnarray}
G(a,d)=\left(\frac{a}{d}\right)G(1, d),
\end{eqnarray}
 where $\left(\frac{a}{d}\right)$ is  the Jacobi symbol. 
Moreover, 

\begin{equation}
G(1,d)=
             \begin{cases}
             \sqrt{d}, &  d\equiv 1\ (\mod 4)   \\
              \i\sqrt{d}, &  d\equiv 3\ (\mod 4).  
             \end{cases}
\end{equation}

(2) If $d=2^k$, then for  $k\geq 4$, 
\begin{eqnarray}
G(a, 2^k)=2G(a, 2^{k-1}).
\end{eqnarray}

(3) If $d=bc$ with $\gcd(b,c)=1$, then 
\begin{eqnarray}
G(a, bc)=G(ab,c)G(ac, b).
\end{eqnarray}

\section{\texorpdfstring{Half Gauss sum for $\xi_d=-\omega_{2d}$ with even $d$}{Half Gauss sum for xid=-omega2d with even d}}
\label{appen:xi_ev}
In the main text, we chose $\xi_d=\omega_{2d}$ for all even numbers $d$.
Note that in the case when $d$ is even, $\xi_d$ can be chosen to
be $\pm \omega_{2d}$. Here, we consider the case 
$\xi_d=-\omega_{2d}$ for all even numbers $d$. To distinguish these two cases, 
we define $G_{1/2}(a, d)_+$ for the case when $\xi_d=\omega_{2d}$ and 
$G_{1/2}(a, d)_-$ for the case when $\xi_d=-\omega_{2d}$ for even $d$. Thus, we have the following two properties for 
$G_{1/2}(a,d)_{-}$.

\begin{lem}
 If $d$ is even, then 
\begin{eqnarray}
G_{1/2}(a, d)_-=G_{1/2}(a(N_1+bN_2),b)_-G_{1/2}(aN_2, c),
\end{eqnarray}
where $d=bc$, $\gcd(b,c)=1$,  $2|b$ and integers $N_1$ and $N_2$ satisfy $N_1c+N_2b=1$.

\end{lem}
\begin{proof}
Following the approach in the proof of Proposition \ref{thm:HG1}, we obtain
 \begin{eqnarray*}
 \xi^{ax^2}_d
 =(-1)^{ax^2}\omega^{aN_1x^2}_{2b}
\omega^{aN_2x^2}_{2c}
&=&(-1)^{ay^2}
\omega^{aN_1y^2}_{2b}
\xi_{c}^{aN_2z^2}
(-1)^{aN_2y^2}\\
&=&(-\omega_{2b})^{a(N_1+bN_2)y^2}
\xi_{c}^{aN_2z^2}\\
&=&\xi^{a(N_1+bN_2)y^2}_{b}
\xi_{c}^{aN_2z^2},
\end{eqnarray*}
which completes the proof of the lemma.
\end{proof}

\begin{lem}
 If $m\geq 3$, then 
 \begin{eqnarray}
 G_{1/2}(a, 2^m)_-
 = 2G_{1/2}(a, 2^{m-2})_+.
 \end{eqnarray}
\end{lem}
\begin{proof}
For $m\geq 3$,
\begin{eqnarray*}
 G_{1/2}(a, 2^{m})_-&=&
 \sum_{x\in [2^m]}
 (-\omega_{2^{m+1}})^{ax^2}\\
& =&\sum_{x\in [2^{m-1}]}
 \left[(-\omega_{2^{m+1}})^{ax^2} + (-\omega_{2^{m+1}})^{a(x+2^{m-1})^2} \right]\\
 &=&\sum_{x\in [2^{m-1}]}(-\omega_{2^{m+1}})^{ax^2} 
 \left[1+ (-\omega_{2^{m+1}})^{a2^mx+a2^{2m-2}}\right]\\
 &=&\sum_{x\in [2^{m-1}]}(-1)^{ax^2}\omega^{ax^2}_{2^{m+1}} 
 [1+(-1)^{x}]\\
 &=& \sum_{y\in [2^{m-2}]}
 \omega^{a(2y)^2}_{2^{m+1}} [1+(-1)^{2y}]\\
&=&2 \sum_{y\in [2^{m-2}]}
 \omega^{4ay^2}_{2^{m+1}} 
 =2 \sum_{y\in [2^{m-2}]}
 \omega^{ay^2}_{2^{m-1}} \\
 &=&2G_{1/2}(a, 2^{m-2})_+.
 \end{eqnarray*}
\end{proof}

\section{Relationship between half Gauss sums and zeros of a polynomial}
\label{app:zeros}

In this appendix, we explore the relationship between half Gauss sums $Z_{1/2}(d,f)$ and the number of zeros of functions of the form $f(x)-k \ (\mod d)$ or $(\mod 2d)$. We start with the following theorem.

\begin{thm}
Let $f:\mathbb Z_d^n \rightarrow \mathbb Z$.
\begin{enumerate}
    \item If $d$ is even, then
    \begin{align}
    \label{eq:FourierTransformEven}
        \left|\left\{x \in \mathbb{Z}_{d}^{n}: f(x)=j \bmod 2 d\right\}\right|=\frac{1}{2 d} \sum_{k=0}^{2 d-1} \xi_{d}^{-kj} Z_{1/2}(d, k f).
    \end{align}
    \item If $d$ is odd, then
    \begin{align}
    \label{eq:FourierTransformOdd}
        \left|\left\{x \in \mathbb{Z}_{d}^{n}: f(x)=j \bmod d\right\}\right|=\frac{1}{ d} \sum_{k=0}^{d-1} \xi_{d}^{-kj} Z_{1/2}(d, k f).
    \end{align}    
\end{enumerate}
\end{thm}
\begin{proof} \hfill
\begin{enumerate} 
    \item If $d$ is even, then $\xi_d=\omega_{2d}$ and $\xi^{2d}_d=1$. Hence,
    \begin{align}
    Z_{1/2}(d, k f)
    &=
    \sum_{x\in \mathbb Z_d^n} \xi_d^{k f(x)} \nonumber\\
    &=\sum^{2d-1}_{j=0}
    \xi^{kj}_d\ \left|\left\{x \in \mathbb{Z}_{d}^{n}: f(x)=j \bmod 2 d\right\}\right| .
    \end{align}
    By taking the inverse Fourier transform, we obtain Eq.~\eqref{eq:FourierTransformEven}.
    
    \item If $d$ is odd, then $\xi^d_d=1$.
    Hence,
    \begin{align}
    Z_{1/2}(d, kf)
    &=
    \sum_{x\in \mathbb Z_d^n} \xi_d^{k f(x)} \nonumber\\
    &=\sum^{d-1}_{j=0}
    \xi^{kj}_d\ \left|\left\{x \in \mathbb{Z}_{d}^{n}: f(x)=j \bmod  d\right\}\right| .
    \end{align}
    By taking the inverse Fourier transform, we obtain Eq.~\eqref{eq:FourierTransformOdd}.
\end{enumerate}
\end{proof}

This allows us to write the number of zeros of a function $f:\mathbb Z_d^n \rightarrow \mathbb Z \ \bmod d$
in terms of half Gauss sums:
\begin{thm}
\begin{align}
    \label{eq:FourierTransform}
        \left|\left\{x \in \mathbb{Z}_{d}^{n}: f(x)=0 \bmod  d\right\}\right|=\frac{1}{d} \sum_{l=0}^{d-1} Z_{1/2}(d, s_d l f) ,
    \end{align}
    where $s_d = 2$ if $d$ is even and $1$ if $d$ is odd.
\end{thm}
\begin{proof}
When $d$ is odd, setting $j=0$ in Eq.~\eqref{eq:FourierTransformOdd} gives Eq.~\eqref{eq:FourierTransform}.

Next, let $d$ be even. Then,
\begin{align}
    \left|\left\{x \in \mathbb{Z}_{d}^{n}: f(x)=0 \bmod  d\right\}\right| 
    &=
    \left|\left\{x \in \mathbb{Z}_{d}^{n}: f(x)=0 \bmod  2d\right\}\right| \nonumber\\
    &\quad +
    \left|\left\{x \in \mathbb{Z}_{d}^{n}: f(x)=d \bmod  2d\right\}\right|  \nonumber\\
    &= \frac{1}{2 d} \sum_{k=0}^{2 d-1} Z_{1/2}(d, k f)+\frac{1}{2 d} \sum_{k=0}^{2 d-1} \xi_{d}^{-k d} Z_{1/2}(d, kf) \nonumber \\
    &= \frac{1}{2 d} \sum_{k=0}^{2d-1}\left(1+(-1)^{k}\right) Z_{1/2}(d, k f)\nonumber\\
    &= \frac{1}{d} \sum_{l=0}^{d-1} Z_{1/2}(d, 2 l f),
\end{align}
where we used $\xi^d_d = -1$ in the third line.

\end{proof}

\section{Characterization of periodic polynomials of degree \texorpdfstring{$\leq 3$}{3} for \texorpdfstring{$d=2$}{d=2}}
\label{sec:periodicPoly3d2}

In this appendix, we give a characterization of polynomials with degree $\leq 3$ that satisfy the periodicity condition for $d=2$. We will use the following notation: let $\modtwo(x)$ be the unique integer $y\in \mathbb Z_2$ for which $x\equiv y \bmod 2$.

We start by proving the following identity.
\begin{lem}
\label{lem:identitymod}
Let $a,b,c,x,y \in \mathbb Z$. If $a$, $b$ and $c$ have the same parity (i.e. if $a$, $b$ and $c$ are either all even or all odd), then
\begin{align}
\label{eq:identitymod}
a x^{2} y + b x y^{2}+c x y &\equiv a \  \modtwo(x^{2} y)+b \ \modtwo(x y^{2})+c\ \modtwo(x y) \pmod 4 \\
&= (a+b+c) \modtwo(xy) \pmod 4 .
\end{align}
\end{lem}
\begin{proof}
Write $x=2q+u$ and $y = 2r+v$, where $q,r \in \mathbb Z$ and $u,v \in \mathbb Z_2$.

Then the RHS of Eq.~\eqref{eq:identitymod} is
\begin{align*}
\mathrm{RHS} &=a \ \modtwo(x^{2} y)+b \ \modtwo\left(x y^{2}\right)+c \ \modtwo(x y) \\
&=a \ \modtwo\left[(2 q+u)^{2}(2 r+v)\right]+b \ \modtwo \left[(2 q+u)(2 r+v)^{2}\right] \\
&\qquad +c \ \modtwo [(2 q+u)(2 r+v)] \\
&=a \ \modtwo\left(u^{2} v\right)+b \ \modtwo\left(u v^{2}\right)+c \ \modtwo(u v) \\
&=a u v+b u v+c u v \\
&=(a+b+c) u v \\
&= (a+b+c) \modtwo(x)\modtwo(y) \\
&= (a+b+c) \modtwo(xy).
\end{align*}

On the other hand, the LHS of Eq.~\eqref{eq:identitymod} is
\begin{align*}
\mathrm{LHS}
&= a x^{2} y+b x y^{2}+c x y
\\
&=
a(2 q+u)^{2}(2 r+v)+b(2 q+u)(2 r+v)^{2}+c(2 q+u)(2 r+v) \\
&= a\left(4 q^{2}+4 q u+u^{2}\right)(2 r+v)+b(2 q+u)\left(4 r^{2}+4 r v+v^{2}\right) \nonumber\\&\quad
+c(4 q r+2 q v+2 u r+u v)
\\
&\equiv au^{2}(2 r+v)+b(2 q+u) v^{2}+c(2 q v+2 u r+u v) \quad \bmod 4 \\
&=2 a u^{2} r+a u^{2} v+2 b q v^{2}+u b v^{2}+2 c q v+2 c u r+c u v\\
&=2 a u r+a u v+2 b q v+u b v+2 c q v+2 c u r+c u v \quad \because u, v \in\{0,1\}
\\
&= 2(a+c) u r+2(b+c) q v+(a+b+c) u v \\
&\equiv (a+b+c)uv \ \mod 4 ,
\end{align*}
where the last equivalence holds because $a,b,c$ have the same parity, i.e. $a+c$ and $b+c$ are even.

\end{proof}

\begin{thm}
Let $n\in \mathbb Z^+$ and let
\begin{align*}
g\left(x_{1}, \ldots, x_{n}\right)=\sum_{1 \leqslant i_1 \leqslant i_2 \leqslant i_3 \leqslant n} a_{i_1 i_2 i_3} x_{i_1} x_{i_{2}} x_{i_{3}}+\sum_{1 \leqslant i_{1} \leqslant i_2 \leqslant n} a_{i_1 i_{2}} x_{i_1} x_{i_{2}}+\sum_{i=1}^{n} a_{i} x_{i} + a
\end{align*}
be a polynomial of degree $\leq 3$ with coefficients that satisfy
$a_{ijk}$, $a_{ij}$, $a_i$ and $a \in \mathbb Z$ for all $i\leqslant j\leqslant k \in \{1,\ldots,n\}$. Then, $g$ satisfies the periodicity condition for $d=2$ if and only if for all distinct $i, j, k \in \{1,\ldots,n\}$,
\begin{enumerate}
    \item[(i)] $a_i$, $a_{iii}$ and $a_{ijk}$ are even,
    \item[(ii)] $a_{ij}$, $a_{ijj}$ and $a_{iij}$ have the same parity.
\end{enumerate}
\label{thm:degree3characterization}
\end{thm}

\begin{proof} \hfill
\begin{enumerate}
    \item[$(\Rightarrow)$]
Assume that $g$ satisfies the periodicity condition for $d=2$. Then for all $x_1,\ldots,x_n \in \mathbb{Z}$,
\begin{align}
    & \i^{g\left(x_{1}, \ldots, x_{n}\right)} 
    = 
    \i^{g\left(x_{1} \bmod 2, \ldots,  x_{n} \bmod 2\right)} \\
    \iff \quad &
    g(x_{1}, \ldots, x_{n}) \equiv g(x_{1} \bmod 2, \ldots, x_{n} \bmod 2) \ \bmod 4 .
\label{eq:periodicityConditiong2}
\end{align}

Denote 
\begin{align}
\tilde{g}(x, y, z) &= 
g(x,y,z,0,\ldots,0) \\
&=
a_{111} x^{3}+a_{222} y^{3}+a_{333} z^{3} \nonumber\\
&\quad +a_{112} x^{2} y+a_{113} x^{2} z+a_{122} x y^{2} \nonumber\\
&\quad +a_{223} y^{2} z+a_{133} x z^{2}+a_{233} y z^{2} +a_{123} x y z \nonumber\\
&\quad+a_{11} x^{2}+a_{22} y^{2}+a_{33} z^{2} \nonumber\\
&\quad+a_{12} x y+a_{13} x z+a_{23} y z \nonumber\\
&\quad+a_{1} x+a_{2} y+a_{3} z+a .
\end{align}

Then, Eq.~\eqref{eq:periodicityConditiong2} implies that for all $x,y,z\in\mathbb Z$,
\begin{align}
\label{eq:gtildemod2condition}
    \tilde g(x,y,z) = \tilde g(x \bmod 2, y \bmod 2, z \bmod 2) \bmod 4 .
\end{align}

We will now use Eq.~\eqref{eq:gtildemod2condition}
repeatedly to find necessary conditions that the coefficients of the polynomial $g$ must satisfy.

First, Eq.~\eqref{eq:gtildemod2condition} implies that
\begin{align}
\tilde{g}(0,0,0) &=\tilde{g}(2,0,0) \bmod 4 \\
\implies\quad
0 &=a_{111} 2^{3}+a_{11} 2^{2}+a_{1} 2 \bmod 4 \nonumber\\
&= 2 a_{1} \bmod 4 ,
\end{align}
which implies that $a_1$ is even. By symmetry between 1 and $i$ for $i\in \{1,\ldots,n\}$,
\begin{align}
\label{eq:aieven}
    \boxed{
    a_i \mbox{ is even } \quad \forall i \in \{1,\ldots,n\}.
    }
\end{align}

Second, Eq.~\eqref{eq:gtildemod2condition} implies that
\begin{align}
&\quad \tilde{g}(1,0,0) =\tilde{g}(-1,0,0) \bmod 4 \\
&\implies\quad
a_{111}+a_{11}+a_1 =-a_{111} +a_{11} -a_{1} \bmod 4 \\
&\implies\quad
2a_{111}+2a_1 = 0 \bmod 4 .
\end{align}
    
By Eq.~\eqref{eq:aieven}, $a_1$ is even, and so $2a_1=0 \bmod 4$. Hence,
\begin{align}
    2a_{111} = 0 \bmod 4 ,
\end{align}
which implies that $a_{111}$ is even. By symmetry between 1 and $i$ for $i\in \{1,\ldots,n\}$,
\begin{align}
\label{eq:aiiieven}
    \boxed{
    a_{iii} \mbox{ is even } \quad \forall i \in \{1,\ldots,n\}.
    }
\end{align} 
  
Third, Eq.~\eqref{eq:gtildemod2condition} implies that
\begin{align*}
& \hspace{2.5em} \tilde{g}(0,1,0) 
=\tilde{g}(2,1,0) \bmod 4 \\
&\implies
 a_{222}+a_{22}+a_2 = a_{111} 8 + a_{222} + a_{112} 4 + a_{122} 2 \\
& \hspace{9em}+ a_{11} 4 + a_{22} + a_{12} 2 + a_1 2 + a_2
\bmod 4 \nonumber\\
& \implies
2a_1 + 2a_{12} + 2a_{122} = 0 \bmod 4 .
\end{align*}
By Eq.~\eqref{eq:aieven}, $a_1$ is even, and so, $2a_1=0 \bmod 4$. Hence,
\begin{align}
    2(a_{12}+a_{122}) = 0 \bmod 4,
\end{align}
which implies that $a_{12}+a_{122}$ is even, i.e.~$a_{12}$ and $a_{122}$ have the same parity.

By symmetry,
\begin{align}
\label{eq:aiiiparity}
    \boxed{
    a_{ij}, a_{ijj}, a_{iij} \mbox{ have the same parity } \quad \forall i<j \in \{1,\ldots,n\}.
    }
\end{align}

Fourth, Eq.~\eqref{eq:gtildemod2condition} implies that
\begin{align*}
& \hspace{2.5em} \tilde{g}(0,1,1) 
=\tilde{g}(2,1,1) \bmod 4 \\
&\implies
 a_{222}+a_{333}+a_{223}
 +a_{233}+a_{22}+a_{33}
 +a_{23}+a_2+a_3 
\\
& \hspace{4.5em} = a_{111} 8 + a_{222} + a_{333} +a_{112} 4 + a_{113} 4 + a_{122} 2 + a_{223} \\ 
&\hspace{6em} + a_{133} 2
+ a_{233} + a_{123} 2 + a_{11}4 +a_{22}+a_{33} + a_{12} 2 \\
&\hspace{6em} + a_{13} 2 + a_{23} + a_1 2 + a_2 + a_3 \bmod 4 \\
&\implies 
2(a_{122}+a_{133}+a_{123}+a_{12}+a_{13}+a_1) = 0 \bmod 4.
\end{align*}

By Eq.~\eqref{eq:aiiieven},
$a_{122}+a_{12}$ and $a_{133}+a_{13}$ are both even, and hence
$2(a_{122}+a_{12}) = 0 \bmod 4$ and 
$2(a_{133}+a_{13}) = 0 \bmod 4$. Also, 
Eq.~\eqref{eq:aieven} implies that $a_1$ is even, and so, $2a_1=0 \bmod 4$. Therefore,
\begin{align}
    2a_{123} = 0 \bmod 4,
\end{align}
which implies that $a_{123}$ is even. By symmetry,
\begin{align}
\label{eq:aijkeven}
    \boxed{
    a_{ijk} \mbox{ is even } \quad \forall i<j<k \in \{1,\ldots,n\}.
    }
\end{align} 

Together, Eqs.~\eqref{eq:aieven}, 
\eqref{eq:aiiieven},  \eqref{eq:aiiiparity} and \eqref{eq:aijkeven} imply
the consequent of the logical biconditional in Theorem 
\ref{thm:degree3characterization}.

\item[$(\Leftarrow)$] 
Assume that (i) and (ii) in Theorem \ref{thm:degree3characterization} hold. Then,
\begin{align}
&\qquad    g(\modtwo(x_1),\ldots,\modtwo(x_n)) \nonumber\\ &= 
\sum_{1 \leqslant i_1 \leqslant i_2 \leqslant i_3 \leqslant n} a_{i_i i_2 i_3} \modtwo(x_{i_1}) \modtwo(x_{i_{2}}) \modtwo(x_{i_{3}}) \nonumber\\ &\quad +
\sum_{1 \leqslant i_{1} \leqslant i_2 \leqslant n} a_{i_1 i_{2}} \modtwo(x_{i_1}) \modtwo(x_{i_{2}})+\sum_{i=1}^{n} a_{i} \modtwo(x_{i}) + a
\nonumber\\
&=
\sum_{1 \leqslant i_1 \leqslant i_2 \leqslant i_3 \leqslant n} a_{i_i i_2 i_3} \modtwo(x_{i_1} x_{i_{2}} x_{i_{3}}) \nonumber\\ &\quad +
\sum_{1 \leqslant i_{1} \leqslant i_2 \leqslant n} a_{i_1 i_{2}} \modtwo(x_{i_1} x_{i_{2}})+\sum_{i=1}^{n} a_{i} \modtwo(x_{i}) + a \nonumber\\
&= \sum_i \underbrace{a_{iii} \modtwo(x_i^3)}_{\Circled{1}} + \sum_{i<j} 
\underbrace{\left[  a_{iij} \modtwo(x_i^2 x_j) +
a_{ijj} \modtwo(x_i x_j^2) +
a_{ij} \modtwo(x_i x_j)\right]
}_{\Circled{2}}
 \nonumber\\
&\quad +
\sum_{i<j<k} \underbrace{a_{ijk} \modtwo(x_i x_j x_k)}_{\Circled{3}}
+ \sum_i \underbrace{a_{ii} \modtwo(x_i^2)}_{\Circled{4}}
+ \sum_i \underbrace{
a_i \modtwo(x_i)}_{\Circled{5}} + a
\label{eq:gmod2xi}
\end{align}

To evaluate \Circled{1}, \Circled{3} and \Circled{5} $\mod 4$, we use the identity
\begin{align}
    2a \ \modtwo(x) \equiv 2ax \pmod 4 \quad \forall a,x \in \mathbb Z
\end{align}

Since $a_{iii}$, $a_{ijk}$ and $a_i$ are even,
it follows that 
\begin{align}
    \Circled{1} &= a_{iii} x_i^3 \bmod 4 
    \label{eq:circled1} \\
    \Circled{3} &= a_{ijk} x_i x_j x_k \bmod 4
    \label{eq:circled3}
    \\
    \Circled{5} &= a_{i} x_i \bmod 4 \label{eq:circled5}
\end{align}

To evaluate \Circled{3} $\mod 4$, we use the identity
\begin{align}
    a \ \modtwo(x^2) \equiv ax^2 \pmod 4 \quad \forall a,x \in \mathbb Z
\end{align}
which gives
\begin{align}
\label{eq:circled4}
    \Circled{4} = a_{ii} x_i^2 \bmod 4
\end{align}

To evaluate \Circled{2} $\mod 4$, we use Lemma \ref{lem:identitymod}, which implies that
\begin{align}
\label{eq:circled2}
    \Circled{2} = a_{iij} x_i^2 x_j +
a_{ijj} x_i x_j^2 +
a_{ij} x_i x_j \ \mod 4
\end{align}

Substituting Eqs.~\eqref{eq:circled1}, \eqref{eq:circled2}, \eqref{eq:circled3}, \eqref{eq:circled4} and \eqref{eq:circled5}
into Eq.~\eqref{eq:gmod2xi} gives
\begin{align}
    g(\modtwo(x_1),\ldots,\modtwo(x_n)) = g(x_1,\ldots,x_n)
\end{align}
which means that $g$ satisfies the periodicity condition.

\end{enumerate}

\end{proof}

\end{appendix}

\bibliographystyle{ieeetr}
\bibliography{bib}{}

\end{document}